\documentclass[journal,twoside,web]{ieeecolor}
\usepackage{generic}
\usepackage{cite}
\usepackage{amsmath,amssymb,amsfonts}
\usepackage{algorithmic}
\usepackage{graphicx}
\usepackage{textcomp}

\makeatletter
\def\fnum@figure{\textcolor{subsectioncolor}{\sf Fig.~\thefigure}}
\def\fnum@table{\textcolor{subsectioncolor}{\sf TABLE~\thetable}}
\makeatother

\def\BibTeX{{\rm B\kern-.05em{\sc i\kern-.025em b}\kern-.08em
    T\kern-.1667em\lower.7ex\hbox{E}\kern-.125emX}}

\usepackage{graphicx, color, xcolor} % for pdf, bitmapped graphics files
\usepackage{epsfig} % for postscript graphics files
\usepackage{amsmath, bbm, bm} 
\usepackage{amsfonts, amssymb}%special font
\usepackage{mathrsfs}
\usepackage{anyfontsize}  
\usepackage{framed}

\usepackage{tikz}

\usepackage[framemethod=TikZ]{mdframed}
\usepackage{mathtools}

\mdfdefinestyle{theoremstyle}{backgroundcolor=black!7,roundcorner=5pt,outerlinewidth=1.1}

\mdfdefinestyle{corostyle}{roundcorner=5pt}
\mdfdefinestyle{propstyle}{roundcorner=5pt}
\mdfdefinestyle{lemmastyle}{roundcorner=5pt}
\mdfdefinestyle{defnstyle}{roundcorner=5pt}

\newcommand {\cA}{{\mathcal{A}}}
\newcommand {\cB}{{\mathcal{B}}}

\newcommand {\bX} {\mathbf{X}}
\newcommand {\bY} {\mathbf{Y}}
\newcommand {\bW} {\mathbb{W}}

\newcommand {\bx} {{\bf x}}
\newcommand {\by} {{\bf y}}

\newcommand {\onu}{\overline{\nu}}
\newcommand {\bpi} {\boldsymbol{\pi}}

\newcommand{\Eqdef}{:=}

\newcommand {\N} {\mathbb{N}}
\newcommand {\R} {{\rm I\kern-2.5pt R}}
\newcommand {\C} {{\rm I\kern-5pt C}}

\newtheorem{lemma}{Lemma}

\newtheorem{coro}{Corollary}

\newtheorem{remark}{Remark}
\newtheorem{defn}{Definition}%[chapter]
\newcommand{\beqa}{\begin{eqnarray}}
\newcommand{\eeqa}{\end{eqnarray}}
\newcommand{\beqan}{\begin{eqnarray*}}
\newcommand{\eeqan}{\end{eqnarray*}}
\newcommand{\beq}{\begin{equation}}
\newcommand{\eeq}{\end{equation}}

\newtheorem{thm}{Theorem}[section]
\newcommand{\bfl}{\begin{flushleft}}
\newcommand{\efl}{\end{flushleft}}
\newcommand{\bgnv}{\begin{verbatim}}
\newcommand{\ednv}{\end{verbatim}}

\newcommand{\myb}{\hspace{-0.1in}}
\newcommand{\myeq}{& \hspace{-0.1in} = & \hspace{-0.1in}}

\newcommand{\lb}{\nonumber \\}
\newcommand{\myarr}{\begin{array}{lll}}

\newcommand{\myleq}{& \myb \leq & \myb}

\newcommand{\bitem}{\begin{itemize}}
\newcommand{\eitem}{\end{itemize}}
\newcommand{\benum}{\begin{enumerate}}
\newcommand{\eenum}{\end{enumerate}}

\newcommand{\norm}[1]{\left| \left| #1 \right| \right|}

\newcommand{\bPP}[2]{{\mathbbm P}_{#1}\left[ #2 \right]}
\newcommand{\myhb}{\hspace{-0.3in}}

\newcommand{\myskip}{\\ \vspace{-0.1in}}
\newcommand{\mydef}{\overset{def}{=}}

\newcommand{\oby}{\overline{{\bf y}}}
\newcommand{\os}{\overline{s}}
\newcommand{\ow}{\overline{w}}
\newcommand{\opi}{\overline{\pi}}
\newcommand{\obP}{\overline{\bf P}}
\newcommand{\obY}{\overline{\bf Y}}

\DeclareMathOperator{\EX}{\mathbb{E}}% expected value

\def\QED{~\rule[-1pt]{5pt}{5pt}\par\medskip}
\let\labelindent\relax

\usepackage{enumitem}
\usepackage{chngcntr}
\title{\LARGE \bf
Stabilizing a Queue Subject to Activity-Dependent Server Performance
}

\author{Michael Lin \and Richard J. La \and Nuno C. Martins \thanks{The authors are with the 
Department of Electrical \& Computer 
Engineering and 
the Institute for Systems Research, the
University of Maryland, College Park, MD 20742. 
Email: \{mlin1025, hyongla, nmartins\}@umd.edu.}
\thanks{This work is supported in part by 
AFOSR Grant FA95501510367 and NSF Grant ECCS 1446785.}
}

\begin{document}

\maketitle
\thispagestyle{empty}
\pagestyle{plain}

%%%%%%%%%%%%%%%%%%%%%%%%%%%%%%%%%%%%%%%%%%%%%%%%%%%%%%%%%%%%%%%%%%%%%%%%%%%%%%%%
\begin{abstract}
We consider a discrete-time system comprising a first-come-first-served queue, a non-preemptive server, 
and a scheduler that governs the assignment of tasks in the queue to the server. 
%New tasks arrive at the queue according to a Bernoulli process. 
The server has an {\em availability} state that indicates, at each instant, whether the server is busy working 
on a task or is available. In the latter case, if the queue is nonempty, then a task-assignment control policy implemented by the scheduler
either assigns a new task to the server or allows it to rest. The server also has an integer-valued {\em 
activity state} that is non-increasing 
during rest periods, and is non-decreasing otherwise. 
An instantaneous service rate function ascribes to each 
possible value of the activity state a probability
that the server can complete a task in one time step. 
For a typical instantaneous service rate function, the completion probability decreases (server performance worsens) as the activity state increases. 
The scheduler policy 
has access to the queue size and the entire state of 
the server. 

In this article, we study the problem of designing scheduler policies that stabilize the queue. We show that 
stability, whenever viable, can be achieved by a 
simple policy that bases its decisions on 
the availability state, a threshold applied to the activity state, and a flag that indicates when the 
queue is empty. The supremum of the service 
rates achievable by stabilizing policies can be determined by a finite 
search. Our results remain valid even when the instantaneous service rate function is not monotonic. 
\end{abstract}

%INTRODUCTION%%%%%%%%%%%%%%%%%%%%%%%%%%%%%%%%%%%%%%%%%%%%%%%%%%%%%%%%%%%%%%%%%%%%%%%%%%%%%%%
% * <hyongla@gmail.com> 2017-01-26T17:29:36.671Z:
%
% ^.

\section{Introduction}

Recent advances in information and communication 
networks, and sensor technologies, fostered a new wide
range of applications. These 
include networked systems in which 
sensors are equipped with communication modules
 powered by renewable energy, such as 
solar and geothermal energy. Invariably, operation of these
systems requires the management of queued tasks. Consequently, a 
queueing paradigm is needed that goes beyond
the classical case in which the 
performance of the server is time-invariant. 
Notably, in these new applications,
server performance often depends on its activity, understood as a proxy for the cumulative effect on resources of the history of task assignments. Typically, the performance of the server degrades when resources get depleted as a result of higher activity, and may recover when its activity is lowered to allow for the renewal of resources.

This article starts by proposing a queuing system in which 
the performance of the server depends on its activity. 
Subsequently, we study
two (soon to be stated) problems associated with the design of task-assignment
control policies for the scheduler. These policies will govern the assignment of tasks from the queue to the server with the goal of keeping the queue stable. A major challenge in the design of these policies is that they will have to manage the trade-off between assigning tasks immediately, which may decrease server performance due to the decrease of resource levels associated with increased activity, or wait until resource levels are renewed at the expense of allowing the queue length to increase as new tasks arrive at a fixed rate. 

\noindent {\bf Convention:} Throughout this article, we will refer to
the task-assignment control policy at the 
scheduler as \underline{scheduler policy} 
or, simply, as \underline{policy} for short when appropriate. 
%In particular, we are interested in identifying policies with a simple structure, which is optimal in that it keeps the task queue stable whenever doing so is possible using some policy.

To be more specific, we consider a queueing system 
comprising the following three components:
\\ \vspace{-0.12in}

\noindent $\bullet$ 
A first-in first-out unbounded {\bf queue} registers a new 
task when it arrives and removes it as soon as work on it is 
completed. It has an internal state, its queue 
size, which indicates the number of uncompleted tasks in 
the queue.

\noindent $\bullet$ 
The {\bf server} performs the work required by each 
task assigned to it. It has an internal state with two 
components. The first is the {\em availability} 
state, which indicates whether the server is currently working on a task or is available to 
start a new one. We assume that the server is non-preemptive, 
which in our context means 
that the server gets busy when it starts work on a new task, 
and it becomes available again only after the task is 
completed. The second component of the state is finite-valued and is termed 
{\em activity}. It accounts for the cumulative effect of the history of task assignments on resources that influence the performance of the server. Such an activity 
state could, for instance, account for the battery charge 
level of an energy harvesting module that powers the server~(see~Section~\ref{subsec:Example}) 
or the status of arousal or fatigue of a human operator that assists 
the server or supervises the work. In our framework, the activity state plays a central role in that it determines the server performance understood as the probability 
that, within a 
given time-period, the server can complete a task. Namely, a 
decrease in performance causes an increase in the expected 
time needed to service a task. 

\noindent $\bullet$
The {\bf scheduler} has access to the queue size
and the entire state of the server. When the server is 
available and the queue is not empty, the scheduler policy 
decides whether to assign a new task or to 
allow for a rest period. Our formulation admits 
non-work-conserving policies that may choose to assign rest periods even 
when the queue is not empty. This would allow the server to 
rest as a way to steer the activity 
state towards a range that can deliver better 
long-term performance.
\\ \vspace{-0.1in}

We adopt a stochastic discrete-time framework in which time 
is uniformly partitioned into epochs, within which new tasks 
arrive according to a Bernoulli process. The probability of 
arrival per epoch is termed arrival rate\footnote{Notice 
that, unlike the nomenclature we adopt here, {\em arrival 
rate} is commonly used in the context of Poisson arrival 
processes. This distinction is explained in detail in 
Section~\ref{sec:ArrivalProcess}.}. We constrain our 
analysis to {\em stationary scheduler policies} 
that are invariant under epoch shifts. We discuss 
our assumptions and provide a detailed description of our 
framework in Section~\ref{sec:ProbFormModel}. 

\subsection{Main Problems}
The following are the main challenges studied in this 
article:

\noindent \rule{\columnwidth}{1pt}
\noindent {\bf Main Problems:}
%\remind{Should we use the same notation used in the second paper to define the problems?}
\begin{itemize} 
\item[P-i)] An arrival rate is qualified as stabilizable 
when there is a stationary scheduler policy that stabilizes the queue. Given 
a server, we seek to compute the supremum of the set of 
stabilizable arrival rates.

\item[P-ii)] We seek to propose scheduler policies that 
have a simple structure and are guaranteed to stabilize 
the queue for any stabilizable arrival rate.
\\ \vspace{-0.25in}
\end{itemize} 
\rule{\columnwidth}{1pt}
\\ \vspace{-0.1in}

Notice that, as alluded to above in the scheduler 
description, we allow non-work-conserving scheduler 
policies. This means that, in addressing Problem P-i), 
we must allow policies that are a function of 
not only the queue size, but also the activity 
and availability states of the server. The fact that 
this functional dependence complicates the design of the policies
justifies the importance of addressing Problem P-ii). 

\vspace{-.1 in}
\subsection{Motivating Example}
\label{subsec:Example}
The need to analyze a class of remote wireless monitoring systems~\textendash~each comprising a sensor that measures and broadcasts to a base station variations on a bridge joint's strain or temperature~\textendash~was a motivation for the work reported in this article. We proceed to describe in detail the components of one such system in the context of our framework: \vspace{.03 in} \\
\noindent {\bf Arrival process:} A new packet recording a strain or temperature measurement is created and logged into a queue when the absolute value of the difference between the current measurement and the last recorded measurement exceeds a predefined threshold. Hence, decreasing (increasing) the threshold increases (decreases) the rate at which new packets are created and entered into the queue. \vspace{.03 in} \\ \noindent {\bf Tasks:} The transmission from the sensor to the base station (or gateway) of the packet at the head of the queue is a task. \vspace{.03 in} \\ \noindent {\bf Server:} According to our framework, the server for this example is the communication link consisting of a wireless transmitter at the sensor and a receiver at the base station (or gateway). When requested by the scheduler, the server will repeatedly attempt to complete the task (transmit the measurement or packet at the head of the queue) until it succeeds\footnote{We assume that, every time it successfully receives a packet, the base station can use enough power to guarantee that an acknowledgement is provided to the sensor.}. The transmitter at a sensor operates on batteries that are recharged by a solar panel. Discharge curve characteristics~\cite{Shepherd1965Design-of-prima,Chen2006Accurate-electr}, as well as power management firmware on the sensor, reduces (resp. increases) the transmit power when the state of charge of the battery decreases (resp. increases). This is consequential because the probability of a fading event causing a failed transmission diminishes (resp. grows) when the transmit power increases (resp. decreases). The difference between the capacity of the battery and its current state of charge is the activity state.  Hence, the probability that a transmission attempt will succeed (complete the task) is a function of the activity state, as envisaged in our framework. \vspace{.03 in} \\ \noindent {\bf Scheduler:} A microcontroller at the sensor implements a policy that decides based on the activity state, the queue size and availability of the server whether to initiate the transmission of the packet at the head of the queue (new task).

%\vspace{-.25 in}
\subsection{Preview of Results}

The following are our main results (presented in Section~\ref{sec:Mainresults}) and their correspondence 
with Problems~P-i) and~P-ii). 

\ \noindent R-i) In Theorem~\ref{thm:necessity}, we show that the 
supremum mentioned in P-i) can be computed by maximizing a function 
whose domain is the {\em finite} set of activity 
states. The fact that such a quantity can be determined by a 
finite search is not evident because the queue size 
is unbounded. 

\ \noindent R-ii) As we state in Theorem~\ref{thm:sufficiency}, 
given a server, there is a threshold policy that stabilizes 
the queue for any stabilizable arrival rate. The threshold 
policy assigns a new task only when the server is available, 
the queue is not empty, and the activity state is 
less than a threshold chosen to be the value (found by a
{\em finite} search) at which the maximum referred to in 
R-i) is attained. This is our answer to Problem P-ii).

From this discussion we conclude that, to the extent that 
the stability of the system is concerned, 
we can focus solely on the threshold policies outlined 
in R-ii). It is also important to observe that, as we discuss 
in Remark~\ref{rem:results} of Section
\ref{sec:Mainresults}, Theorem~\ref{thm:sufficiency} 
is valid even when the performance of the server is not 
monotonic with respect to the activity state. 

\vspace{-.1 in}
\subsection{Related Literature}
	\label{subsec:Related}

%\remind{Hello Richard, the concept of time-varying service rate may be confusing. How about time-varying instantaneous service rate, or time-varying service time, or time-varying efficiency?}

The stability of queueing systems has been 
studied extensively %, using various techniques
%and different definitions of stability
\cite{BaccelliBremaud, El-Taha99, MeynTweedie, Dai20}, 
and there exists a large volume of literature.
Earlier studies of queueing 
systems with time-varying 
parameters date back to the studies by 
Conway and Maxwell \cite{ConwayMaxwell62}, 
Jackson~\cite{Jackson63}, Yadin an Naor
\cite{Yadin63}, Gupta~\cite{Gupta67} 
and Harris~\cite{Harris67}.
A summary of earlier studies on queues with 
state-dependent parameters can be found in  
\cite{Dshalalow} and references therein. 
The manuscript \cite{El-Taha99} discusses
{\em rate stability}, which is 
different from the usual definition of
stability in stochastic literature, and 
proves that if the conditional output rate of 
the server when the queue size is sufficiently
large exceeds the input rate, the
system is rate stable.

In addition, with the rapid growth of
wireless networks, there has been 
much interest in understanding and
designing efficient scheduler 
policies with time-varying  
channel conditions that 
affect the probability of 
successful transmissions~\cite{Agrawal02, 
Andrews04, Borst05}. 
Many studies focus on designing throughput 
optimal scheduler policies that can stabilize
the system for any arrival rate
(vector) that lies in the stability 
region (e.g., \cite{Andrews04, Ren2004Optimal-transmi, 
Tassiulas1992Stability-prope}). 
However, there is a major difference 
between these studies and our study. 
In wireless networks, channel conditions 
and probability 
of successful transmission/decoding vary
independently of 
the scheduling decisions chosen by 
the scheduler. 
In our study, on the other hand, the 
probability of successfully completing
a task within an epoch depends on the 
past history of scheduling 
decisions. Consequently, the current 
scheduling decision 
affects the future efficiency of the
server.  

In \cite{Ramjee96, 
Fang02}, the authors study the problem of
designing admission control policies to 
ensure quality-of-service in the form of
call block probabilities. In particular, 
\cite{Bekker2006Optimal-Admissi} considers
queueing systems with workload-dependent
service rates and studies the problem of 
devising an admission control policy that 
maximizes the long-run throughput. It 
shows that, under the assumption
that the expected return time to given 
workload is unimodal, there exists an
optimal admission control 
policy with a threshold on workload. 
Despite some similarity between this study 
and ours, there are major
differences: first, the studied problems are
different; unlike in 
\cite{Bekker2006Optimal-Admissi}, in our
problem we do not control the workload 
arriving at the queue, and the service
rate depends on the activity state of the
server that summarizes its recent utilization level.
We are interested in designing a non-work-conserving 
scheduling policy that permits a non-preemptive 
server to rest after completing a task, 
instead of forcing it to operate 
inefficiently. Second, perhaps
more noteworthy, we demonstrate that
there is an optimal threshold policy
(on the activity state) without assuming
that the service rate is unimodal in 
the activity state.

Another area that is closely related to
our study is task scheduling for human
operators/servers. 
The performance and management of human
operators and servers (e.g., bank tellers, 
toll collectors, doctors, nurses, emergency
dispatchers) has been the subject
of many studies in the past, e.g., 
\cite{Borghini14, Edie54, Shunko17, 
YerkesDodson}. Recently, with 
rapid advances in information and sensor
technologies, human supervisory control, 
which requires processing a large amount of
information in a short period, potentially
causing information overload, became an 
active research area~\cite{Supervisory, 
Sheridan1997Handbook-of-hum}.

As human supervisors play a critical role in 
the systems (e.g., supervisory control
and data acquisition), 
there is a resurging interest in understanding
and modeling the performance of humans
under varying settings.
Although this is still an active research area, 
it is well documented that the performance
of humans depends on many factors, including
arousal and perceived workload
\cite{Edie54, Asaro07, KcTerwiesch09, Shunko17}. 
For example, the well-known Yerkes-Dodson law
suggests that moderate levels of arousal are
beneficial, leading to the inverted-U model
\cite{YerkesDodson}.

In closely related
studies, Savla and Frazzoli
\cite{Savla2010Maximally-stabi, Savla2012A-Dynamical-Que}
investigated the problem of designing
a task release control policy. They 
assumed periodic task arrivals and modeled
the dynamics of server utilization, which 
determines the service time
of the server, using a 
differential equation; the server
utilization increases when the server
is busy and decreases when it is idle. 
They showed that, when all tasks bring 
identical workload, a policy that allows
a new task to be released to the 
server only when 
its utilization is below a suitably
chosen threshold, is maximally stabilizing
\cite[Theorems III.1 and 
III.2]{Savla2012A-Dynamical-Que}. 
Moreover, they proved that when task workloads
are modeled using independent and 
identically distributed %(i.i.d.) 
random variables, the maximum achievable
throughput increases
compared to the homogeneous workload
cases. 
%
%We note
%that a similar model and methodology 
%are applicable to studying and designing 
%process scheduling policies for  
%multi-core processors with adaptive 
%control of clock speed or voltage of each 
%core subject to power and thermal 
%constraints~\cite{Rao08}.

We note that existing results in the 
literature (e.g., \cite{Baccelli94}) may 
be used to prove Theorem~\ref{thm:sufficiency}. 
However, we elect to provide a more direct 
proof in the paper for two reasons: first, 
a self-contained proof eliminates the need 
to introduce additional setup and notation to 
make use of existing results
and, in our opinion, improves readability. Second, 
we believe that our approach and reported 
auxiliary results in the proof can be used more 
readily in extending
the results of this paper. Finally, 
to the best of our knowledge, the more notable
result in Theorem~\ref{thm:necessity}
does not follow 
from any existing results in the literature.

\noindent {\bf Paper Structure:} A stochastic discrete-time model is described in Section~\ref{sec:ProbFormModel}. In it we also introduce notation, key concepts and propose a Markov Decision Process (MDP) framework that is amenable to performance analysis and optimization. Our main results are discussed in Section~\ref{sec:Mainresults}, and Section~\ref{Sec:auxiliaryMDP} describes an auxiliary MDP that is central to our analysis. The proofs of our results are presented in Section~\ref{SectionMainProofs}, and Section~\ref{sec:Conclusion} provides concluding remarks.

\section{Stochastic Discrete-Time Framework}
	\label{sec:ProbFormModel}
	
In the following subsection, we describe a discrete-time 
framework that follows from assumptions on when the states 
of the queue and the server are updated and how actions are 
decided. In doing so, we will also introduce the notation 
used to represent these discrete-time processes. A 
probabilistic description that leads to a tractable MDP 
formulation is deferred to Section~\ref{sec:ProbModel}.

\vspace{-.1 in}
\subsection{State Updates and Scheduler Operation: Timing 
	and Notation}
\label{sec:timingnotation}

We consider an infinite horizon problem in which the {\em 
physical} (continuous) time set is $\mathbb{R}_+ 
\Eqdef [0, \infty)$, which we 
partition uniformly 
into half-open intervals of positive duration $\Delta$ as 
$ \mathbb{R}_+ = \cup_{k=0}^{\infty} \  [ \ k\Delta, \ (k+1) \Delta \ ) $. Each interval is called an {\em epoch}, and epoch $k$ refers 
to $[k\Delta,(k+1)\Delta )$. Our formulation and results are 
valid regardless of the epoch duration $\Delta$. 
We reserve $t$ to 
denote continuous time, and $k$ is the discrete-time
index we use to represent epochs.

\begin{figure}
\centerline{
	\includegraphics[width=2.8in]{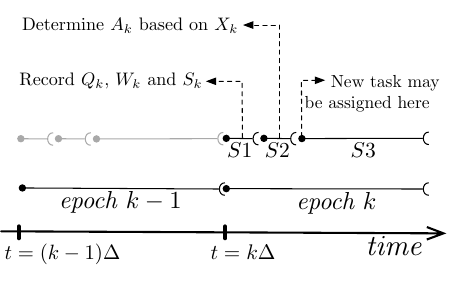}
}
	\caption{Illustration of time uniformly divided into 
	epochs and when updates and actions are taken. 
	(Assuming $k \geq 1$)}
	\label{fig:Figtimings}
	\vspace{-.2 in}
\end{figure}

Each epoch is subdivided into three half-open subintervals 
denoted by stages $S1$, $S2$ and $S3$ (see Fig.
\ref{fig:Figtimings}). As we explain below, stages $S1$ 
and $S2$ 
are allocated for basic operations of record keeping, 
updates and scheduling decisions. Although, in practice, the 
duration of these stages is a negligible fraction of $\Delta
$, we discuss them here in detail to clarify the causality 
relationships among states and actions. We also introduce 
notation used to describe certain key discrete-time 
processes that are indexed with respect to epoch number.

In addition, 
even though more general arrival processes
can be handled, for simplicity of exposition, 
we assume
that new task arrivals occur in accordance 
with a Bernoulli process, and the Bernoulli 
random variables are assumed independent. Thus, 
at most one new task arrives at queue during
each epoch. We denote the number of
tasks that arrive during epoch $k$ by $B_k$, which
takes values in $\{0, 1\}$. 

Furthermore, during each epoch, the scheduler
policy assigns at most one task to the server. 
Hence, the server either works on a single task 
or remains idle at any given time. We denote the
number of tasks that the server completes
during epoch $k$ by $D_k$, 
which takes values in $\{0, 1\}$.
\subsubsection{\bf \underline{Stage $S1$}} 

The following updates take place during stage $S1$ of 
epoch $k+1$:

% We assume that at most one new task arrives 
% during each epoch. In addition, the server 
% can work on at most one task at any given 
% time and, within each epoch, the scheduler 
% assigns at most one task to the server. 
% However, these assumptions are not critical
% to our main findings and can be relaxed
% to allow more general arrival distributions
% and batch processing. 

The {\em queue size} at time $t=k \Delta$ is denoted by 
$Q_k$, and it is updated
according to a Lindley's equation~\cite{lindley1952theory}:
\beqa
Q_{k+1}
= \max \big\{ 0, Q_k + B_k - D_k \big\}, \ 
	k \in \N \Eqdef \{0, 1, \ldots\}
	\label{eq:Lindley}
\eeqa

The {\em availability state of the server} at time 
$t=k\Delta$ is 
denoted by $W_{k}$ and takes values in 
$\mathbb{W} \Eqdef \{ \mathcal{A}, \mathcal{B} \}$.
We use $W_{k}=\mathcal{B}$ to indicate that the server is 
busy working on a task at time $t=k \Delta$. If it is 
available at time $t=k \Delta$, then $W_k=\mathcal{A}$. 
The update mechanism for $W_k$ is as follows:

\begin{itemize}

\item If $W_{k} = \mathcal{A}$, then $W_{k+1} = \mathcal{A}$ 
when either no new task was assigned during epoch $k$, or a 
new task was assigned and completed during epoch $k$. If 
$W_{k} = \mathcal{A}$ and a new task is assigned during 
epoch $k$ which is not completed until $t=(k+1)\Delta$, then 
$W_{k+1} = \mathcal{B}$.

\item If $W_{k} = \mathcal{B}$ and the server completes the 
task by time $t=(k+1)\Delta$, then $W_{k+1} = \mathcal{A}
$. Otherwise, $W_{k+1} = \mathcal{B}$.

\end{itemize}

We use $S_k$ to denote the {\em activity 
state} at time $t=k \Delta$, and assume that it takes 
values in
$\mathbb{S} \Eqdef \{1,\ldots,n_s\}$.
The activity state is non-decreasing while the 
server is working and is non-increasing when it is 
idle. In Section~\ref{sec:ProbModel}, we describe an 
MDP that specifies probabilistically how $S_k$ transitions 
to $S_{k+1}$, conditioned on whether the server worked 
or rested during epoch $k$.

Without loss of generality, we assume that $Q_k$, $W_k$ and 
$S_k$ are initialized as follows:
$ Q_0=0$, $W_0 = \mathcal{A}$, and $S_0=1 $.
The overall state of the server is represented compactly by 
$\bY_k$, which takes 
values in $\mathbb{Y}$, defined as follows:
\begin{equation*} 
\bY_k \Eqdef (S_k,W_k), \qquad
\mathbb{Y} \Eqdef \mathbb{S} \times \mathbb{W}
\end{equation*}
In a like manner, we define the overall state for the MDP taking values
in $\mathbb{X}$ as follows:
\begin{equation*} 
\bX_k \Eqdef (\bY_k,Q_k), \qquad \mathbb{X} \Eqdef \mathbb{S}\times \Big( (\mathbb{W}\times\mathbb{N})  \diagdown  
(\mathcal{B},0) \Big)
\end{equation*}
From the definition of $\mathbb{X}$, it follows that 
when the queue is empty, there is no task for the server 
to work on and, hence, it cannot be busy.

\subsubsection{\underline{{\bf Stage} $S2$}} 

It is during stage $S2$ of epoch $k$ that the scheduler 
issues a decision based on $\bX_{k}$, in accordance
with the employed scheduler policy: 
let $\mathbb{A} \Eqdef \{ \mathcal{R}, 
\mathcal{W} \}$ represent the set of possible actions 
that the scheduler can request from the server, where $
\mathcal{R}$ and 
$\mathcal{W}$ represent `rest' and `work', respectively. 
The assumption that the server is non-preemptive and the fact that 
no new tasks can be assigned when the queue is empty, lead to the 
following set of admissible actions for each possible state 
${\bx=(s,w,q)}$ in $\mathbb{X}$: 
\begin{equation}
\label{ActionConstraints}
\mathbb{A}_\bx
= \begin{cases}
	\{\mathcal{R}\} & \text{if  $q = 0, \  $ {\small (impose 
		`rest' when queue is empty)}} \\
    \{\mathcal{W}\} & \text{if  $q > 0$ and  $w = \cB$, \ {\small 
    	(non-preemptive server)}} \\
    \mathbb{A} & \text{otherwise}. \\
	\end{cases}
\end{equation} 

We denote the action chosen by the adopted 
scheduler policy at epoch $k$ by 
$A_k$, which takes values in $\mathbb{A}_{\mathbf{X}_k}$.

Notice that, in light of the restrictions in~(\ref{ActionConstraints}), only when the queue is nonempty and the server is available at time $k$ will the scheduler policy have the authority to decide whether $A_k = \mathcal{W}$ (assign a new task for the server to work on) or $A_k = \mathcal{R}$ (let the server rest).

As we discuss in Section~\ref{subsec:StateEvolution}, we 
focus on the design of stationary policies that determine 
$A_k$ as a function of $\bX_k$. 

\subsubsection{\underline{{\bf Stage} $S3$}} A task can arrive at any time during each epoch, but we assume that work on a new task can be assigned to the server only at the beginning of stage $S3$. More specifically, the scheduler acts as follows:

\begin{itemize}
\item If $W_k = \mathcal{A}$ and $A_k = \mathcal{W}$, then 
the server starts working on a new task at the head of the
queue when stage $S3$ of epoch $k$ begins.

\item When $W_k = \mathcal{A}$, the scheduler can also 
select $A_k = \mathcal{R}$ to signal that no work will be 
performed by the server during the remainder of epoch. Once 
this `rest' decision is made, a new task can be assigned no 
earlier than the beginning of stage $S3$ of epoch $k+1$. 
Since the scheduler is non-work-conserving, it may decide to 
assign such `rest' periods as a way to possibly reduce $S_{k
+1}$ and to improve future performance. 

\item If $W_k = \mathcal{B}$, the server was still 
working on a task at time $t=k\Delta$. In this case, because 
the server is non-preemptive, the scheduler picks $A_k = 
\mathcal{W}$ to indicate that work on the current task is 
ongoing and must continue until it is completed and no new 
task can be assigned during epoch $k$.

\end{itemize}

\vspace{-.1 in}
\subsection{State Updates and Scheduler Operation: Probabilistic Model}
\label{sec:ProbModel}

Based on the formulation outlined in Section~
\ref{sec:timingnotation}, we proceed to
describe a discrete-time MDP that models
how the states of the server and queue evolve over time for 
any given scheduler policy.

\subsubsection{\underline{Arrival Process}}
\label{sec:ArrivalProcess} 

We assume that tasks arrive during each epoch according 
to a Bernoulli process $\{B_k : k \in \N\}$. 
The probability 
of arrival for each epoch (~${P(B_k=1)}$~) is called the 
{\em arrival rate} 
and is denoted by $\lambda$, which is 
assumed to belong to~$(0, 1)$. 
Although we assume Bernoulli arrivals 
to simplify our analysis and discussion, 
more general arrival distributions 
(e.g., Poisson distributions) can be
handled only with minor changes as it
will be clear. 

Notice that, as we discuss in Remark 1 below, 
our nomenclature for $\lambda$ should not be confused 
with the standard definition of arrival rate for Poisson 
arrivals. Since our results are valid irrespective of 
$\Delta$, including when it is arbitrarily small, 
the remark also gives a sound justification
for our adoption of the Bernoulli arrival model by 
viewing it as a {\em discrete-time approximation}
of the widely used Poisson arrival model.

\begin{remark} It is a well-known fact that, as $\Delta$ tends to zero, a Poisson process in continuous time $t$, with arrival rate $\tilde{\lambda}$, is arbitrarily well approximated by $B_{\lfloor t/\Delta \rfloor}$ with $\lambda = \Delta \tilde{\lambda}$.
\end{remark}

\subsubsection{\underline{Activity-Dependent Server Performance}}
  
In our formulation, the efficiency or performance of the 
server during an epoch is modeled with the help of 
an {\em instantaneous service rate function} $\mu: \mathbb{S} \to 
(0, 1)$. More specifically, if the server works on a task
during epoch $k$, \underline{the probability} that it 
completes the task by the end of the epoch is $\mu(S_k)$. 
This holds irrespective of whether the task is newly 
assigned or inherited as ongoing work from a previous epoch.\footnote{This assumption is
introduced to simplify the exposition. However, 
more general scenarios
in which the probability of task completion
within an epoch depends on the total 
service received by the task prior to 
epoch $k$ can be handled by extending the
state space and explicitly
modeling the total service received by 
the task in service.} 
Thus, $\mu$ quantifies the effect 
of the activity state on the performance of the 
server. {\bf The results presented throughout this 
article are valid for {\em any} choice of $\mu$ with 
codomain $(0,1)$}.

\subsubsection{\underline{Dynamics of the Activity State}}

We assume that (i)~$S_{k+1}$ is equal to either $S_k$ 
or $S_k+1$ when $A_k$ is 
$\mathcal{W}$ and (ii)~$S_{k+1}$ is either 
$S_k$ or $S_k - 1$ if $A_k$ 
is $\mathcal{R}$.  The state-transition probabilities for $S_k$ are specified below for every $s$ and $s'$ in~$\mathbb{S}$:
\begin{subequations}
\label{Def-SDynamics}
\begin{align}
 P_{S_{k+1} | S_k, A_k}(s' \ | \ s, \mathcal{W}) & =
\begin{cases}
	\rho_{s,s+1} & \mbox{if }  
    	s' = s + 1 \\
    1 - \rho_{s, s+1} & \mbox{if } 
    	s' = s \\
    0 & \mbox{otherwise}
	\end{cases} \\
 P_{S_{k+1} | S_k, A_k}(s' \ | \ s, \mathcal{R}) &= 
\begin{cases}
	\rho_{s,s-1} & \mbox{if }
    	s' = s-1 \\
    1 - \rho_{s, s-1} & \mbox{if } 
    	s' = s \\
    0 & \mbox{otherwise}
	\end{cases}
\end{align}
\end{subequations} where the parameters $\rho_{s,s'}$ quantify the likelihood that the activity 
state will transition to a greater or lesser value, depending on 
whether the action is $\mathcal{W}$ or $\mathcal{R}$, respectively.
Here, we assume that $\{\rho_{s,s+1} : 1\leq s < n_s \}$ and $\{\rho_{s,s-1}: 1< s \leq n_s \}$ take values in $(0,1)$. We also adopt the convention that $\rho_{1,0}=\rho_{n_s,n_s+1}=0$.

\begin{remark}\label{def:justificationMarkovActionDependent} According to~(\ref{Def-SDynamics}), the state of charge (activity state) of the battery for the example in Section~\ref{subsec:Example} would be modeled as a controlled Markov chain, which can be viewed as an approximation of the models in~\cite{Sudevalayam2011Energy-Harvesti,Kansal2007Power-Managemen} and~\cite[Part~IV]{Priya2009Energy-Harvesti}. This follows the approach of the work described in the surveys~\cite{Ulukus2015Energy-harvesti,Leong2018Optimal-control,Jog2019Channels-learni}, in which controlled Markovian models are widely used to characterize the time-evolution of the state of charge of the batteries of energy harvesting devices. As is discussed in~\cite{Jog2019Channels-learni}, there are other processes that could be viewed as the activity state of a system and whose dynamics could be approximated by a controlled Markovian model after appropriate discretization of the state-space. Examples of such processes include the temperature of the wireless transmission module in communication~\cite{Koch2009Channels-that-h,Baknina2018aEnergy-harvesti} and distributed estimation systems~\cite{Forte2013Thermal-aware-s}, or, as is described in Section~\ref{subsec:Related}, the state quantifying human operator fatigue in~\cite{Savla2012A-Dynamical-Que}.
\end{remark}

\begin{figure}
    \centering
    \begin{tikzpicture}[scale=0.8, every node/.style={transform shape}]
		\draw [fill = black!5](0,0) node (v1) {\large $1$} circle (.5);
		\draw [fill = black!5](2,0) node (v2) {\large $2$} circle (.5);
		\draw [fill = black!5](4,0) node (v3) {\large $3$} circle (.5);
		\draw [fill = black!5](7,0) node (v3) {\large $n_s$} circle (.5);
		\node at (5.5,0) {$\hdots$};
		\draw [-latex] plot[smooth, tension=.7] coordinates {(0.5,0.25) (1,0.5) (1.5,0.25)};
		\draw [-latex] plot[smooth, tension=.7] coordinates {(2.5,0.25) (3,0.5) (3.5,0.25)};
		\draw [-latex] plot[smooth, tension=.7] coordinates {(4.5,0.25) (5,0.5)};
		\draw [-latex] plot[smooth, tension=.7] coordinates {(4.25,0.5) (4.25,1) (3.75,1) (3.75,0.5)};
		\draw [-latex] plot[smooth, tension=.7] coordinates {(2.25,0.5) (2.25,1) (1.75,1) (1.75,0.5)};
		\draw [-latex] plot[smooth, tension=.7] coordinates {(0.25,0.5) (0.25,1) (-0.25,1) (-0.25,0.5)};
		\draw [-latex] plot[smooth, tension=.7] coordinates {(7.25,0.5) (7.25,1) (6.75,1) (6.75,0.5)};
		\draw [dashed, -latex] plot[smooth, tension=.7] coordinates {(-0.25,-0.5) (-0.25,-1) (0.25,-1) (0.25,-0.5)};
		\draw [dashed, -latex] plot[smooth, tension=.7] coordinates {(1.5,-0.25) (1,-0.5) (0.5,-0.25)};
		\draw [dashed, -latex] plot[smooth, tension=.7] coordinates {(1.75,-0.5) (1.75,-1) (2.25,-1) (2.25,-0.5)};
		\draw [dashed, -latex] plot[smooth, tension=.7] coordinates {(3.5,-0.25) (3,-0.5) (2.5,-0.25)};
		\draw [dashed, -latex] plot[smooth, tension=.7] coordinates {(3.75,-0.5) (3.75,-1) (4.25,-1) (4.25,-0.5)};
		\draw [dashed, -latex] plot[smooth, tension=.7] coordinates {(5,-0.5) (4.5,-0.25)};
		\draw [dashed, -latex] plot[smooth, tension=.7] coordinates {(6.5,-0.25) (6,-0.5)};
		\draw [dashed, -latex] plot[smooth, tension=.7] coordinates {(6.75,-0.5) (6.75,-1) (7.25,-1) (7.25,-0.5)};
		\draw [-latex] plot[smooth, tension=.7] coordinates {(6,0.5) (6.5,0.25)};
		\draw [-latex](4,-1.5) -- (5,-1.5);
		\draw [-latex, dashed](0.25,-1.5) -- (1.25,-1.5);
		\node[right] at (1.5,-1.5) {$A_{k} = \mathcal{R}$};
		\node[right] at (5.25,-1.5) {$A_{k} = \mathcal{W}$};
	    \node at (3,0.75) {$\rho_{2,3}$};
	    \node at (3,-0.75) {$\rho_{3,2}$};
        \node at (1,0.75) {$\rho_{1,2}$};
        \node at (1,-0.75) {$\rho_{2,1}$};
    \end{tikzpicture}
    \caption{Dynamics of the Activity State $S_k \in \mathbb{S}$.}
    \label{fig:activityState}
    \vspace{-.1 in}
\end{figure}
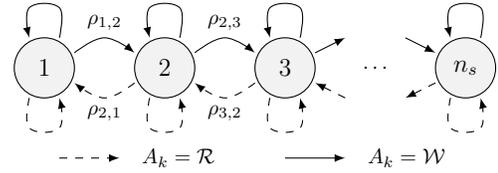

\subsubsection{\underline{Transition probabilities for $\bX_k$}}

We consider that $S_{k+1}$ is independent of {${(W_{k+1}, Q_{k+1})}$} 
when conditioned on {{$(\bX_k, A_k)$}}. Under this assumption, the 
transition probabilities for $\bX_k$ can be written as follows: 
for every $\bx$, $\bx'$ in $\mathbb{X}$ and  $a$ in $\mathbb{A}_{\bx}$,
\beqa
\label{Eqdef:MDP}
&& \myhb P_{\bX_{k+1}|\bX_k,A_k}( \bx'  \ | \ \bx, a) \lb
\myeq  P_{S_{k+1} | \bX_k, A_k}(s' \ | \ \bx, a) \lb 
&& \times P_{W_{k+1},Q_{k+1} | \bX_k, A_k}(w', q' \ | \ 
	\bx, a) \lb 
\myeq P_{S_{k+1} | S_k, A_k}(s' \ | \ s, a) 
	\label{eq:Xt+1} \\
&& \times P_{W_{k+1},Q_{k+1} | \bX_k, A_k}(w', q' \ | \ 
	\bx, a).
	\nonumber
\eeqa

We assume that, within each epoch $k$, the events that
(a) there is a new task arrival during the epoch and (b) a 
task being serviced during the epoch is completed by the end 
of the epoch are independent when conditioned on $\bX_k$ and 
$\{A_k=\mathcal{W}\}$. Hence, the transition probability 
$P_{W_{k+1},Q_{k+1} | \bX_k, A_k}$ in (\ref{eq:Xt+1}) is 
given by the following:
\begin{subequations}
\label{WTransitionDef}
\beqa
&& \myhb P_{W_{k+1},Q_{k+1} | \bX_k, A_k}(w', q' 
	\ | \ \bx, \mathcal{W}) \label{eq:WQW} \\ \nonumber
\myeq \left\{ \begin{array}{ll}
	\mu(s_k) \ \lambda & \mbox{if } w' = \cA \mbox{ and }
    	q' = q, \\
	\mu(s_k) \ (1-\lambda) & \mbox{if } w' = \cA \mbox{ and } 
    	q' = q - 1, \\
	(1-\mu(s_k)) \ \lambda & \mbox{if } w' = \cB \mbox{ and }
    	q' = q + 1, \\
	(1-\mu(s_k)) (1-\lambda) & \mbox{if } w' = \cB 
    	\mbox{ and } q' = q, \\
	0 & \mbox{otherwise, } \\
	\end{array} \right.  \\
%\eeqa
%\beqa
&& \myhb P_{W_{k+1},Q_{k+1} | \bX_k, A_k}(w', q' \ | \ 
	\bx, \mathcal{R}) \label{eq:WQR} \\
\myeq \left\{ \begin{array}{ll}
	\lambda & \mbox{if } w' = \cA \mbox{ and } q' = q+1, \\
	1-\lambda & \mbox{if } w' = \cA \mbox{ and } q' = q, \\
	0 & \mbox{otherwise.} \\
	\end{array} \right.
    \nonumber
\eeqa
\end{subequations}

\vspace{.1 in}
\begin{defn}({\bf MDP $\bX$}) The MDP with input $A_k$ and state $\bX_k$, which at this point is completely defined, is denoted by~$\bX$. 
\end{defn}

Table~\ref{table:Notation} summarizes the notation for MDP~$\bX$.
\begin{table}[h]
\begin{center}
\begin{tabular}{c|l} \hline
$\mathbb{S}$ & set of activity states $\{1, 	
	\ldots, n_s\}$ \\
$\bW \Eqdef \{\cA, \cB\}$ & server availability ($\cA =$ 
	available, $\cB =$ busy) \\
$W_k$ & server availability at epoch $k$ (takes values in 
	$\mathbb{W}$) \\
$\mathbb{Y}$ & server state components $\mathbb{S}\times
	\mathbb{W}$\\
$\bY_k \Eqdef (S_k, W_k)$ & server state at epoch $k$ 
	(takes values in $\mathbb{Y}$) \\
$\mathbb{N}$ & natural number system $\{0,1,2,\ldots \}$. \\
$Q_k$ & queue size at epoch $k$ (takes values in 
	$\mathbb{N}$) \\
$\mathbb{X}$ & state space formed by $\mathbb{S}\times 
	\Big( (\mathbb{W}\times\mathbb{N})\diagdown 
	(\mathcal{B},0) \Big)$ \\
$\mathbf{X}_k \Eqdef (\bY_k, Q_k)$ & system state at 
	epoch $k$ (takes values in $\mathbb{X}$)\\
$\mathbf{X}$ & MDP whose state is $\mathbf{X}_k$ at epoch 
	$k \in \N$ \\
$\mathbb{A} \Eqdef \{\mathcal{R},\mathcal{W}\}$ & possible 
	actions ($\mathcal{R}$ = rest, $\mathcal{W}$ = work) \\
$\mathbb{A_\mathbf{x}}$ & set of actions admissible at a 
	given state $\mathbf{x}$ in $\mathbb{X}$\\
$A_k$ & action chosen at epoch $k$. \\
PMF & probability mass function \\
\hline
\end{tabular}
\vspace{.1in}
\caption{A summary of notation describing MDP $\bX$.}
\label{table:Notation}
\end{center}
\vspace{-0.2in}
\end{table}

%POLICY%
\subsection{Evolution of the System State Under a Stationary Policy}
	\label{subsec:StateEvolution}

We start by defining the class of policies that we consider 
throughout the paper. 

\begin{defn} A stationary randomized policy is specified by
a mapping $\theta: \mathbb{X} \to [0, 1]$ that determines the 
probability that the server is assigned to work on a task or rest, 
as a function of the system state, according to
\begin{align*}
P_{A_k|\mathbf{X}_k,\ldots,\mathbf{X}_0}(\mathcal{W}|x_k,\ldots,x_0) 
	&= \theta(x_k) \ \mbox{ and } \\
P_{A_k|\mathbf{X}_k,\ldots,\mathbf{X}_0}(\mathcal{R}|x_k,\ldots,x_0) 
	&= 1-\theta(x_k).
\\ \vspace{-0.1in}
\end{align*}
\end{defn}

Hence, for given a realization $x_k$ of the state at epoch $k$, the stationary randomized policy specified by a map $\theta: \mathbb{X} \to [0, 1]$ assigns the action $A_k=\mathcal{W}$ with probability $\theta(x_k)$, and $A_k=\mathcal{R}$ with probability $1-\theta(x_k)$. According to~\cite{Puterman2005Markov-decision}, these policies would be further qualified as memoryless because, given $x_k$, they do not rely on previous realizations of the state.
 
\begin{defn}	\label{def:PhiR}
The set of admissible stationary randomized policies satisfying 
(\ref{ActionConstraints}) is denoted by $\Theta_R$.  
\end{defn}

\noindent {\bf Convention:} We adopt the convention that, unless stated 
otherwise, a positive arrival rate $\lambda$ is pre-selected and 
fixed throughout the paper. Although the statistical properties of 
$\bX$ and associated quantities subject to a 
given policy depend on $\lambda$, we simplify our notation by not 
labeling them with $\lambda$.

From (\ref{eq:Xt+1}) - (\ref{eq:WQR}), we conclude that $\mathbf{X}$ 
subject to a policy $\theta$ in  $\Theta_R$ evolves according to a
time-homogeneous Markov chain (MC), which we denote by
$\bX^\theta = \{\bX^\theta_k : k \in \N \}$.
Also, when it is clear from the context, 
we refer to $\bX^\theta$ as {\em the system}. 

%STABILITY%

There are several different but
related definitions of stability for queueing
systems (e.g., rate stability and asymptotic 
average stability), and we refer a reader 
to manuscripts 
\cite{BaccelliBremaud, El-Taha99, MeynTweedie, Dai20}
for more details. 
For our study, we adopt the definition of system 
stability stated below. As it will be clear, for 
our system, this 
definition is equivalent to a definition 
of stability that is common in the literature on stochastic systems, which 
requires the existence of a limiting probability 
distribution (Lemma~\ref{lemma:UniquePMF}).

\begin{defn}[System stability, stabilizability and 
$\Theta_S(\lambda)$] \label{def:Stability} 
For a given policy $\theta$ in $\Theta_R$, 
the system $\bX^\theta$ is stable if it satisfies the following 
properties: 

\noindent (i)~The number of transient states is finite and, hence, 
there is at least one recurrent communicating class (RCC). \\
\noindent (ii)~All RCCs are positive recurrent. 

An arrival rate $\lambda$ is said to be stabilizable when there is a policy $\theta$ in $\Theta_R$ for which $\bX^\theta$ is stable. We also define $\Theta_S(\lambda)$ to 
be the set of randomized policies in $\Theta_R$ that
stabilize the system for a stabilizable arrival rate $\lambda$.
\end{defn}

Before we proceed, let us point out a useful fact
under any stabilizing policy $\theta$ in 
$\Theta_S(\lambda)$. 
 
\begin{lemma}	\label{lemma:UniquePMF}
A stable system $\bX^\theta$ has a unique 
positive recurrent communicating class (PRCC), 
which is aperiodic. Hence, 
there is a unique 
stationary probability mass function (PMF) 
for $\bX^\theta$. 
\end{lemma}
\begin{proof}
Please see Appendix~\ref{appen:UniquePMF} for a proof. 
\end{proof}

\begin{defn} \label{def:pi}
Given a fixed arrival 
rate $\lambda > 0$ and a stabilizing policy 
$\theta$ in $\Theta_S(\lambda)$, we denote the unique 
stationary PMF and PRCC
of $\bX^\theta$ by $\bpi^{\theta} = 
(\pi^{\theta}(\bx) : \bx \in \mathbb{X})$ 
and $\mathbb{C}_\theta$, respectively.
\end{defn}

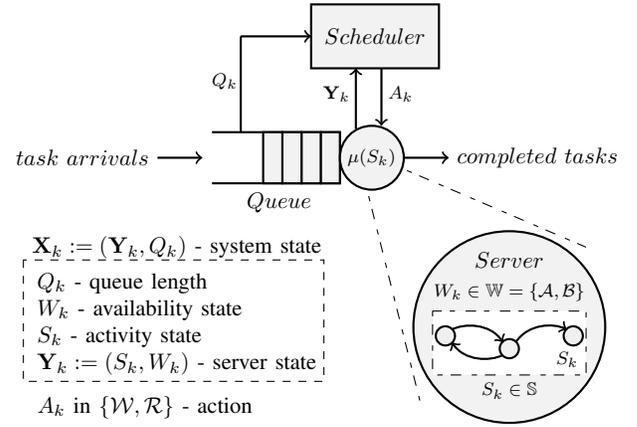
\begin{figure}
    \centering
    \begin{tikzpicture}[scale=0.85, every node/.style={transform shape}]
        %Scheduler
        \draw[thick, fill = black!5]  (0.5,1.8) rectangle (2.5,0.8);
        \node[align=center] at (1.5,1.3) {$Scheduler$};
        
        %Queue
        \draw[black, thick] (-1.05,-0.2) -- ++(2,0) -- ++(0,-0.8) -- ++(-2,0);
        \draw[thick, fill = black!5]  (0.65,-0.2) rectangle +(0.3,-0.8);
        \draw[thick, fill = black!5]  (0.35,-0.2) rectangle +(0.3,-0.8);
        \draw[thick, fill = black!5]  (0.05,-0.2) rectangle +(0.3,-0.8);
        \draw[thick, fill = black!5]  (-0.25,-0.2) rectangle +(0.3,-0.8);
        \draw[<-,thick] (-1.2,-0.6) -- +(-20pt,0) node[left] {$task\ arrivals$};
        \node[align=center] at (0,-1.3) {$Queue$};
        
        %Server
        \draw[black, thick, fill = black!5] (1.45,-0.6) circle [radius=0.5] node {\footnotesize $\mu(S_k)$};
        \draw[->,thick] (1.95,-0.6) -- +(20pt,0) node[right] {$completed\ tasks$};
        \draw[black, thick, fill = black!5] (3.6,-3.25) circle [radius=1.5];
        \draw [dash pattern=on 2pt off 3pt on 4pt off 4pt,fill=white] (2.4,-3) rectangle (4.8,-4);
        \node[anchor=north] at (3.6,-4) {\footnotesize $S_k \in \mathbb{S}$};
        \node[anchor=south east] at (4.8,-4) {\footnotesize $S_k$};
        \node[align=center] at (3.6,-2.2) {$Server$};
        \draw[black, thick, fill=black!6] (2.6,-3.3884) circle (0.15);
        \draw[->,thick] (2.7,-3.2884) .. controls (2.9,-3.1884) and (3.3,-3.1884) .. (3.5,-3.4884);
        \draw[->,thick] (3.5,-3.6884) .. controls (3.3,-3.7884) and (2.9,-3.7884) .. (2.7,-3.4884);
        \draw[black, thick, fill=black!6] (3.6,-3.5884) circle (0.15);
        \draw[black, thick, fill=black!6] (4.6,-3.3884) circle (0.15);
        \draw[->,thick] (3.7,-3.4884) .. controls (3.9,-3.1884) and (4.3,-3.1884) .. (4.5,-3.2884);

        %Connections(Queue to Scheduler)
        \draw[->,thick] (-0.6,1.3) -- (0.5,1.3);
        \draw [black, thick](-0.6,-0.2) -- (-0.6,1.3);
        \node at (-0.85,0.6) {\footnotesize$Q_k$};
        
        %Connections(Scheduler to Server)
        \draw[->,thick] (1.2,-0.15) -- (1.2,0.8);
        \draw[->,thick] (1.6,0.8) -- (1.6,-0.1);
        
        %Legends:
        \node[align=left] at (-1.6,-3.2) {
        $Q_k$ - queue length\\
        $W_k$ - availability state\\
        $S_k$ - activity state \\
        $\bY_k:=(S_k,W_k)$ - server state};
        \node at (0.9,0.4) {\footnotesize$\bY_k$};
        \node at (1.9,0.4) {\footnotesize$A_k$};
        
        \node [align=left] at (-2.1,-4.5) {
        $A_k$ in $\{\mathcal{W}, \mathcal{R}\}$ - action};
        \draw [dashed] (-4,-2.2) rectangle (0.7,-4.1);
        \node [align=left] at (-1.6,-2) { $\bX_k:=(\bY_k,Q_k)$ - system state};
        
        \draw [dash pattern=on 2pt off 3pt on 4pt off 4pt](1.4,-1.2) -- (2.2,-4.4);
        \draw [dash pattern=on 2pt off 3pt on 4pt off 4pt](2,-0.8) -- (4.8,-2);
        \node at (3.6,-2.7) {\footnotesize $W_k \in \mathbb{W}=\{\mathcal{A},\mathcal{B}\}$};
    \end{tikzpicture}
    \caption{System Diagram.}
    \label{fig:System}
    \vspace{-0.1in}
\end{figure}

%\subsection{Comparison to existing work and %%literature review}

%MAIN RESULT%%%%%%%%%%%%%%%%%%%%%%%%%%%%%%%%%%%%%%%%%%%%%%%%%%%%%%%%%%%%%%%%%%%%%%%%%%%%%%%

\section{Main Results}
\label{sec:Mainresults}

Our main results are Theorems~\ref{thm:necessity}~and~
\ref{thm:sufficiency}, where we state that a soon-to-be 
defined quantity $\onu^*$, which can be computed 
efficiently, is the least upper bound of all arrival 
rates for which there exists a stabilizing policy in 
$\Theta_R$ (see Definition~\ref{def:Stability}). The 
theorems also assert that, for any arrival rate $\lambda$ 
less than $\onu^*$, there is a stabilizing deterministic
{\em threshold} policy in $\Theta_R$ with the following structure:

\begin{equation}
\label{eq:ThetaDef}
\theta_{\tau}(s,w,q)
\Eqdef \begin{cases}
	\phi_{\tau}(s,w) & \text{if  $q > 0$}, \\
	0 & \text{otherwise,}
    \end{cases}
\end{equation}
where $\tau$ lies in $\mathbb{S} \cup \{n_s+1\}$, and $\phi_{\tau}$ is a 
threshold policy that acts as follows:
\beqa
\phi_{\tau}(s, w)
\Eqdef \begin{cases}
	0 & \text{if  $s \geq \tau$ and  $w = \cA$}, \\
    1 & \mbox{otherwise.}
	\end{cases}
	\label{eq:ThresholdPolicy}
\eeqa

%\remind{I am not sure that having "CMC S" in the cicle in the figure is helpful.} 
Notice that, when the server is available and the queue is not empty, 
$\theta_{\tau}$ assigns a new task only if $s$ is less than the threshold 
$\tau$ and lets the server rest otherwise. 

In Section~\ref{Sec:auxiliaryMDP}, we introduce an 
auxiliary MDP with {\em finite} 
state space $\mathbb{Y}$, which can be viewed
as the server state in $\bX$ when the 
queue size $Q_k$ is always positive. 
First, using the fact that 
$\mathbb{Y}$ is finite, we demonstrate 
that, for every $\tau$ in {${\mathbb{S} 
\cup \{n_s+1\}}$}, the auxiliary MDP subject 
to $\phi_{\tau}$ has a unique stationary PMF, 
which we denote by $\overline{\bpi}^{\tau}$. 
Then, we show that, for 
any stable system $\bX^\theta$ 
under some policy $\theta$
in $\Theta_R$, we can find a threshold 
policy $\phi_\tau$
for the auxiliary MDP, which achieves the 
same long-term departure rate of completed
tasks as in $\bX^\theta$. As a result, 
the maximum long-term departure rate of 
completed tasks in the auxiliary MDP among
all threshold policies $\phi_\tau$ with 
$\tau$ in {${\mathbb{S} 
\cup \{n_s+1\}}$} serves as an 
upper bound on the arrival rate $\lambda$
for which we can hope to find a stabilizing
policy $\theta$ in $\Theta_R$.

Making use of this observation, 
we define the following important quantity:
\beqa
\onu^* 
\Eqdef \max_{ \tau \in \mathbb{S} \cup \{n_s+1\}} 
	\left( \sum_{(\overline{s}, \overline{w}) \in \mathbb{Y}} 
    	\overline{\pi}^{\tau}(\overline{s}, \overline{w}) 
    		\ \phi_{\tau}(\overline{s}, \overline{w}) 
        	\ \mu(\overline{s}) \right)
	\label{eq:lambda*}
\eeqa
From the definition of the stationary PMF
$\bpi^\tau$, $\onu^*$ can be interpreted as the
maximum long-term departure rate of completed
tasks under any threshold policy of the form 
in (\ref{eq:ThetaDef}), {\em assuming
that the queue is always non-empty}.

\vspace{.1 in}
\underline{The following are the main results of this paper.}
\vspace{.1in}
\begin{thm}\label{thm:necessity} 
(Necessity) If, for a given arrival rate $\lambda$, there 
exists a stabilizing policy in $\Theta_R$, then $\lambda \leq \onu^*$.
\end{thm}
%\noindent \rule{\columnwidth}{1pt}

%\vspace{0.1 in}

A proof of Theorem~\ref{thm:necessity} is provided in Section~\ref{subsec:nec}.

\vspace{0.1 in}

%\noindent \rule{\columnwidth}{1pt} \\ \vspace{-0.25in}
\begin{thm} \label{thm:sufficiency} 
(Sufficiency) Let $\tau^*$ be a maximizer of 
(\ref{eq:lambda*}). If the arrival rate $\lambda$ is 
strictly less than $\onu^*$, then 
$\theta_{\tau^*}$ stabilizes the system.
\\ \vspace{-0.15in}
\end{thm}
 Theorem~\ref{thm:sufficiency} is proved in Section~\ref{subsec:suf}.
\vspace{0.1 in}

%\vspace{0.1 in}

\begin{remark} \label{rem:results}
The following important observations are direct consequences 
of (\ref{eq:lambda*}) and Theorems~\ref{thm:necessity}~and~\ref{thm:sufficiency}:

\begin{itemize}
\item The computation of $\onu^*$ in (\ref{eq:lambda*}) along with a 
maximizing threshold $\tau^*$ relies 
on a {\em finite} search that can be 
carried out efficiently.  

\item The theorems are valid for {\em any} 
choice of $\mu: \mathbb{S} \to 
(0, 1)$. In particular, 
$\mu$ could be multi-modal, increasing or 
decreasing.

\item The search that yields $\onu^*$ 
and an associated $\tau^*$ does not 
require knowledge of $\lambda$.
\end{itemize}
\end{remark}

We point out two key differences
between our study and
\cite{Savla2010Maximally-stabi, Savla2012A-Dynamical-Que}. 
The model employed by Savla and Frazzoli
assumes that the service time function 
is convex, which is analogous to $\mu$ being unimodal
in our formulation. 
In addition, a threshold policy is proved
to be maximally stabilizing only for 
identical task workload. In 
our study, however, we do not impose any
assumption on $\mu$, 
and the workloads of tasks are modeled
using independent and identically distributed random 
variables.\footnote{To be more precise, 
our assumptions correspond to the case
with exponentially distributed workloads. However, 
as mentioned earlier,
this assumption can be relaxed to allow more general
workload distributions.}

\section{An Auxiliary MDP $\overline{\bY}$}
\label{Sec:auxiliaryMDP}

In this section, we describe an \emph{auxiliary} MDP whose state 
takes values in $\mathbb{Y}$ and approximates the server state of $\mathbf{X}$ 
under a subclass of policies in $\Theta_R$, subject to an 
assumption that the queue always has a task to service
whenever the server becomes available. We denote this 
auxiliary MDP by $\overline{\mathbf{Y}}$ and its state at 
epoch $k$ by 
$\overline{\bY}_k=(\overline{S}_k,\overline{W}_k)$ in order to 
emphasize that it takes values in $\mathbb{Y}$. The action chosen 
at epoch $k$ is denoted by $\overline{A}_k$. We use the overline
to denote the auxiliary MDP and any other variables associated 
with it, in order to distinguish them from those of the 
server state in~$\bX$. 

We can view $\overline{\bY}$
as the server state of the 
original MDP $\bX$ for which infinitely many 
tasks are waiting in the queue at the beginning, 
i.e.,  $Q_0 = \infty$. As a result, there is always
a task waiting for service when the server becomes available. 

The reason for introducing $\overline{\bY}$ is the following: 
(i) $\mathbb{Y}$ is finite and, hence, $\overline{\bY}$ is easier 
to analyze than $\bX$, and (ii) we can establish a relation 
between $\bX$ and $\overline{\bY}$, 
which allows us to prove the main results
in the previous section by studying $\overline{\bY}$ instead
of $\bX$. This simplifies the proofs of the theorems
in the previous section. 

\vspace{.1in}
\noindent {\bf Admissible action sets:} 
As the queue size is no longer a component of the state of 
$\overline{\bY}$,  we eliminate the dependence of admissible
action sets on $q$, which was explicitly specified in 
(\ref{ActionConstraints}) for MDP $\bX$, while 
still ensuring that the server is non-preemptive. More 
specifically, the set of admissible actions at each 
element $\overline{\by}=(\overline{s},\overline{w})$ of 
$\mathbb{Y}$ is given by 
%$\overline{\mathbb{A}}_{\overline{w}}$,  
%which is defined as
\begin{equation}
\label{ActionConstraintsauxiliary}
\overline{\mathbb{A}}_{\overline{w}}
\Eqdef \begin{cases}
    \{\mathcal{W}\} & \text{if $\overline{w} = \cB$, \quad 
    	(non-preemptive server)} \\
    \mathbb{A} & \text{if $\overline{w} = \cA$}. \\
	\end{cases}
\end{equation} 
Consequently, for any given realization of the current state $
\overline{\by}_k = (\overline{s}_k,\overline{w}_k)$, $\overline{A}_k$ 
is required to take values in $\overline{\mathbb{A}}_{\overline{w}_k}$.

\vspace{.1 in}
\noindent {\bf Transition probabilities:} 
We define the transition probabilities that specify $\overline{\bY}$, 
as follows:
\beqa
P_{\overline{\bY}_{k+1}|\overline{\bY}_k,\overline{A}_k}( \overline{\by}' 
 	 \ | \ \overline{\by}, \overline{a}) 
& \myb \Eqdef & \myb 
P_{\overline{S}_{k+1} | \overline{S}_k, \overline{A}_k}(\overline{s}' 
	\ | \ \overline{s}, \overline{a}) 
	\label{Eqdef:MDPRED} \\ 
&& \times P_{\overline{W}_{k+1} | 
		\overline{\bY}_k, \overline{A}_k}(\overline{w}' \ |\ \overline{\by}, 
			\overline{a}), 
	\nonumber
\eeqa
where $\overline{\by}$ and $\overline{\by}'$ are in $\mathbb{Y}$, and 
$\overline{a}$ is in $\overline{\mathbb{A}}_{\overline{w}}$. Subject 
to these action constraints, the right-hand terms of (\ref{Eqdef:MDPRED}) 
are defined, in connection with $\bX$, as follows:
\begin{equation}
\label{SDef}
P_{\overline{S}_{k+1} | \overline{S}_k, \overline{A}_k}(\overline{s}' 
	\ | \ \overline{s}, \overline{a}) 
\Eqdef P_{S_{k+1} | S_k, A_k}(\overline{s}' \ | \ \overline{s}, \overline{a}) 
\end{equation}
\begin{subequations}
\begin{align}
P_{\overline{W}_{k+1} | \overline{\bY}_k, \overline{A}_k}(\overline{w}' 
	\ | \ \overline{\by}, \mathcal{W}) 
\Eqdef & \begin{cases} 
	\mu(\overline{s}) & \text{if $\overline{w}'=\mathcal{A}$} \\ 
	1-\mu(\overline{s}) & \text{if $\overline{w}'=\mathcal{B}$} 
	\end{cases} \\
P_{\overline{W}_{k+1} | \overline{\bY}_k, \overline{A}_k}(\overline{w}' 
	\ | \ \overline{\by}, \mathcal{R})
\Eqdef & \begin{cases} 
	1 & \text{if $\overline{w}'=\mathcal{A}$} \\ 
	0 & \text{if $\overline{w}'=\mathcal{B}$} 
	\end{cases}
\end{align}
\end{subequations}

\vspace{.1in}
\noindent {\bf Relating the transition 
probabilities of $\bX$ and~$\overline{\bY}$:} 
From the definition above and (\ref{WTransitionDef}), 
we can show that, for all $q \geq 1$, 
\beqa
&& \myhb P_{\overline{W}_{k+1} | \overline{\bY}_k, \overline{A}_k} 
	\big( \overline{w}' 
	\ | \ \overline{\by}, \mathcal{W} \big) \lb
\myeq \sum_{q' = 0}^\infty P_{W_{k+1},Q_{k+1} | \bX_k, A_k} 
	\big( (\overline{w}',q') \ | \ (\ \overline{\by},q), \mathcal{W} 
		\big),
	\label{equivalenceXandY-W} 
\eeqa
which holds for any $\overline{w}'$ in $\mathbb{W}$ and $\overline{\by}$ 
in $\mathbb{Y}$. Notice that the right-hand side (RHS) of 
(\ref{equivalenceXandY-W}) does not change when we vary $q$ across the 
positive integers. From this, in conjunction with 
(\ref{eq:Xt+1}), (\ref{Eqdef:MDPRED})
and (\ref{SDef}), we have, for all $q \geq 1$, 
\beqa
&& \myhb P_{\overline{\bY}_{k+1} | \overline{\bY}_k, \overline{A}_k} 
	\big(\overline{\by}' \ | \ \overline{\by}, \mathcal{W}) \lb 
\myeq \sum_{q'=0}^{\infty} P_{\bX_{k+1} | \bX_k, A_k} 
	\big( (\overline{\by}',q') \ | \ ( \overline{\by},q), \mathcal{W} 
		\big).
	\label{equivalenceXandY} 
\eeqa
The equality in (\ref{equivalenceXandY}) indicates that 
$P_{\overline{\bY}_{k+1} | \overline{\bY}_k, \overline{A}_k}$ 
also characterizes the transition probabilities of 
the server state $\bY_k = (S_k, W_k)$ in $\bX$ 
when the current queue size is 
positive. This is consistent with our earlier viewpoint that 
$\overline{\bY}$ can be considered the server state in $\bX$ 
initialized with infinite queue length at the beginning. 
We will explore this relationship in Section~\ref{SectionMainProofs}, 
where we use $\overline{\bY}$ to prove Theorems~\ref{thm:necessity}~and~\ref{thm:sufficiency}. 
\\ \vspace{-0.3in}

\subsection{Stationary Policies and 
Stationary PMFs of $\overline{\bY}$}
	\label{subsec:PMFobY}

Analogously to the MDP $\bX$, we only consider 
stationary randomized policies for $\overline{\bY}$, 
which are defined below.  
%\myskip

\begin{defn}[Stationary randomized policies for $\overline{\bY}$] 
We restrict our attention to stationary randomized policies acting on 
$\overline{\bY}$, which are specified by 
a mapping ${\phi: \mathbb{Y} \to [0, 1]}$, 
as follows: for all $k$ in $\N$ and $\overline{\by}_k,\ldots,
\overline{\by}_0$ in $\mathbb{Y}$, 
\beqan
P_{\overline{A}_k | \overline{\bY}_k,\ldots,\overline{\bY}_0 } 
	(\mathcal{W} | \overline{\by}_k,\ldots,\overline{\by}_0)
\myeq \phi(\overline{\by}_k) \\
P_{\overline{A}_k | \overline{\bY}_k,\ldots,\overline{\bY}_0 } 
	(\mathcal{R} | \overline{\by}_k,\ldots,\overline{\by}_0)
\myeq 1- \phi(\overline{\by}_k)
\eeqan
The set of all 
stationary randomized policies for $\overline{\bY}$ 
which honor (\ref{ActionConstraintsauxiliary}) is 
defined to be $\Phi_R$.
\end{defn}

The MDP $\overline{\mathbf{Y}}$ subject to a policy
$\phi$ in $\Phi_R$ is a {\em finite-state} 
time-homogeneous MC and is denoted by 
${\overline{\bY}^{\phi} \Eqdef 
\{\overline{\bY}^{\phi}_k : k \in \N \} }$. 
\vspace{.1 in}
\subsubsection{Recurrent Communicating Classes of $\overline{\bY}^{\phi}$}
Because $\mathbb{Y}$ is finite, for any policy $\phi$ 
in $\Phi_R$, $\overline{\mathbf{Y}}^{\phi}$ has a PRCC and a stationary distribution
\cite{Grimmett2001Probability-and}. In fact, 
there are at most two PRCCs as explained in \underline{Cases~1 and~2} below. 

Define a mapping 
$\mathcal{T}: \Phi_R
\to \mathbb{S} \cup \{0\}$, where 
\beqan
\mathcal{T}(\phi)
\Eqdef \max\{ \os \in \mathbb{S} \ | \ 
	\phi(\os, \cA) = 1 \}, \  \phi \in \Phi_R. 
\eeqan
We assume that $\mathcal{T}(\phi) = 0$ if the
set on the RHS is empty.

\underline{\em Case 1. $\phi(1, \cA) > 0$:} 
First, from the definition of $\mathcal{T}(\phi)$, 
all states $(\os, \ow)$ with $\os
\geq \mathcal{T}(\phi)$ communicate with 
each other, but none of these
states communicates with any other state
$(\os', \ow')$ with $\os' < \mathcal{T}(\phi)$
because $\phi(\mathcal{T}(\phi), \cA) 
= \phi(\mathcal{T}(\phi), \cB) = 1$. 
Second, because $\phi(1, \cA) > 0$ by
assumption, all states 
$(\os', \ow')$
with $\os < \mathcal{T}(\phi)$ 
communicate with states $(\os, \ow)$ with 
$\os \geq \mathcal{T}(\phi)$. 
Together with the first observation, this implies
that these states $(\os', \ow')$ with $\os'
< \mathcal{T}(\phi)$ are transient. Thus, 
there is only one PRCC given by 
\beqa
\mathbb{Y}^\phi
\Eqdef \{ (\os, \ow) \in \mathbb{Y} \ | \ 
	\os \geq \mathcal{T}(\phi) \}. 
	\label{eq:Yphi}
\eeqa

\underline{\em Case 2. $\phi(1, \cA) = 0$:} In 
this case, it is clear that $(1, \cA)$ is an 
absorbing state and forms a PRCC by itself. 
Hence, if $\mathcal{T}(\phi) = 0$, as all 
other states communicate with $(1, \cA)$, 
the only PRCC
is $\{(1, \cA)\}$ and all other states are
transient. On the other hand, if
$\mathcal{T}(\phi) > 1$, for the same reason
explained in the first case, $\mathbb{Y}^{\phi}$
gives rise to a second PRCC, 
and all other states $(\os', \ow')$
with $\os' < \mathcal{T}(\phi)$, except for
$(1, \cA)$, are transient.

In our study, we often limit our discussion to 
randomized policies $\phi$ in $\Phi_R$ with
$\phi(1, \cA) > 0$. For this reason, for 
notational convenience, we define the set of 
randomized policies satisfying this condition 
by $\Phi^+_R$. The reason for this will be
explained in the subsequent section. 

The following proposition is an immediate consequence
of the above observation. 
%\myskip

\begin{coro} 
\label{coro:pibar} For any policy $\phi$ in 
$\Phi^+_R$, 
$\overline{\bY}^{\phi}$ has a unique stationary PMF, 
which we denote by $\overline{\bpi}^{\phi} 
= (\overline{\pi}^{\phi}(\by) : \by \in \mathbb{Y})$.
%\myskip
\end{coro}
\vspace{.1 in}
\subsubsection{Existence Of A Policy For $\overline{\bY}$
With An Identical Steady-State Server State PMF
In $\bX^\theta$} One of key facts which we will make use of 
in our analysis is that, for any stabilizing policy 
$\theta$ in $\Theta_S(\lambda)$, we can find a policy
$\phi$ in $\Phi_R^+$ which achieves the same steady-state
distribution of server state. To this end, we first define,
for each $\overline{\by}$ in $\mathbb{Y}$, 
\beqan
\mathbb{Q}^{\overline{\by}} \Eqdef \{q \in \N \ | \ 
	(\overline{\by},q)\in \mathbb{X} \}.
\eeqan

\begin{defn}[Policy projection map $\mathscr{Y}$] 
Given $\lambda$ in $(0,\onu^*)$, we define a mapping
$\mathscr{Y}: \Theta_S(\lambda) \mapsto \Phi_R$ as follows:
\beqan
\mathscr{Y}(\theta)\Eqdef \phi^{\theta}, \ \theta 
	\in \Theta_S(\lambda), 
\eeqan
where $\phi^{\theta}:\mathbb{Y} \rightarrow [0,1]$ is specified as:
\beqa
\label{QMapDef}
\phi^{\theta}(\overline{\by}) 
\Eqdef \frac{\sum_{q \in \mathbb{Q}^{\overline{\by}}}
	\theta(\overline{\by},q) \pi^{\theta}(\overline{\by},q)}
	{\sum_{q \in \mathbb{Q}^{\overline{\by}}} 
	\pi^{\theta}(\overline{\by},q)}, 
	\ \overline{\by} \in \mathbb{Y}.
\eeqa 
\end{defn}

We first present a lemma that proves useful
in our analysis. 

\begin{lemma}
	\label{lemma:PhiThetaPos}
For every stabilizing policy $\theta$ in 
$\Theta_S(\lambda)$, we have
$\phi^\theta(1, \cA) > 0$.
\end{lemma} 
\begin{proof}
Please see Appendix~\ref{appen:PhiThetaPos}
for a proof. 
\end{proof}

An obvious implication of the lemma is that 
$\mathscr{Y}(\theta)$
belongs to $\Phi^+_R$ for every 
$\theta$ in $\Theta_S(\lambda)$, and there exists
a unique stationary PMF for 
$\overline{\bY}^{\mathscr{Y}(\theta)}$,
namely $\overline{\bpi}^{\mathscr{Y}(\theta)}$.

The following lemma shows that the 
steady-state PMF of the server state 
in $\bX$ under policy $\theta$ in 
$\Theta_S(\lambda)$ is identical to that of 
$\overline{\bY}$ under policy $\mathscr{Y}(\theta)$.
%\myskip

\begin{lemma}
	\label{lemma:PMFequivalence}
Suppose that $\theta \in \Theta_S(\lambda)$. Then, 
we have
\beqa
\overline{\pi}^{\mathscr{Y}(\theta)}(\overline{\by}) 
= \sum_{q\in\mathbb{Q}^{\overline{\by}}}
	\pi^\theta(\overline{\by},q), 
	\ \overline{\by} \in \mathbb{Y}.
	\label{eq:pi-relation}
\eeqa
\end{lemma}
\begin{proof}
A proof is provided in Appendix \ref{appen:PMFequivalence}.
\end{proof}

%MAIN PROOF%%%%%%%%%%%%%%%%%%%%%%%%%%%%%%%%%%%%%%%%%%%%%%%%%%%%%%%%%%%%%%%%%%%%%%%%%%%%%%%
\section{Proofs of The Main Results}
\label{SectionMainProofs}

In this section, we begin with a comment on the 
long-term average departure rate of completed tasks
when the system is stable. Then, we provide the
proofs of Theorems~\ref{thm:necessity} and 
\ref{thm:sufficiency} in Sections~\ref{subsec:nec}~and~\ref{subsec:suf}, respectively.

\begin{remark} 
\label{nuislambda}
Recall from our discussion in Section~
\ref{sec:ProbFormModel}
that, under a stabilizing policy $\theta$ in $
\Theta_S(\lambda)$, there exists a unique stationary 
PMF $\bpi^\theta$. Consequently, the average number of 
completed tasks per epoch converges 
almost surely as $k$ goes to infinity. In other words, 
\beqan
%&& \myhb 
\lim_{k \to \infty} \frac{\sum_{\tau=0}^{k-1} 
	D^{\theta}_\tau}{k} %\lb
\myeq \sum_{\mathbf{x} \in \mathbb{X}} \mu (s) 
\pi^{\theta}(\mathbf{x}) \theta(\mathbf{x}) 
\Eqdef \nu^\theta \ \mbox{almost surely,}
%\mbox{ with probability 1,}
\eeqan 
where $s$, $w$ and $q$ are the coordinates 
of {${\bx = (s,w,q)}$}, and $D^{\theta}_\tau$ is the
number of tasks that the server completes
during epoch $\tau$ under policy $\theta$.
We call $\nu^\theta$ the {\em service rate} of 
$\theta$ (for the given arrival rate $\lambda > 0$). 
Moreover, because $\theta$ is assumed to be a stabilizing 
policy, we have $\nu^\theta = \lambda$. 
\end{remark}

%NECESSARY PROOF%%%%%%%%%%%%%%%%%%%%%%%%%%%%%%%%%%%%%%%%%%%%%%%%%%%%%%%%%%%%%%%%%%%%%%%%%%%%%%%
\subsection{Proof of Theorem~\ref{thm:necessity}}
\label{subsec:nec}

In order to prove Theorem~\ref{thm:necessity}, we make
use of a similar notion of {\em
service rate} of $\overline{\bY}^\phi$, which
can be viewed in most cases as the average number 
of completed tasks per epoch. 
{\bf (Step 1)} We first establish that,
for every stabilizing policy $\theta$, we can 
find a {\em related} policy 
$\phi$ in $\Phi_R$ whose
service rate equals that of $\theta$ or, 
equivalently, the arrival rate $\lambda$. 
{\bf (Step 2)} We prove that $\onu^*$ 
in (\ref{eq:lambda*}) equals the maximum
service rate achievable by any policy in $\Phi_R$. 
Together, they tell us $\lambda \leq \onu^*$. 

\vspace{.1 in}
\noindent {\bf Service rate of
$\overline{\bY}^\phi$:} 
The {\it service rate} 
associated with $\overline{\bY}^{\phi}$ under 
policy $\phi$ in $\Phi_R$ is defined as follows: for 
each $\phi$ in $\Phi_R$, let
$\overline{\boldsymbol{\Pi}}(\phi)$ be the set 
of stationary PMFs of $\overline{\bY}^\phi$.
By Corollary~\ref{coro:pibar}, 
for any $\phi$ in $\Phi^+_R$, there is a unique
stationary PMF and $\overline{\boldsymbol{\Pi}}
(\phi)$ is a singleton.
The service rate of
$\phi$ in $\Phi_R$ is defined to be 
\beqa
\overline{\nu}^{\phi} 
\Eqdef \sup_{\overline{\bpi} \in 
	\overline{\boldsymbol{\Pi}}(\phi)}
	\Big( \sum_{\overline{\by} \in \mathbb{Y}} \mu(\overline{s}) 
	\overline{\pi}(\overline{\by}) \phi(\overline{\by})
	\Big).
	\label{eq:onu} 
\eeqa 
Recall that
$\overline{\by}$ is the pair $(\overline{s},
\overline{w})$ taking values in $\mathbb{Y}$. 
\myskip

\noindent {\bf Step 1:} 
The following lemma illustrates that the long-term service
rate achieved by $\mathscr{Y}(\theta)$ in $\Phi^+_R$ 
equals that of $\theta$. 
%\myskip 

\begin{lemma}
\label{lem:lambdaisnu}
Suppose that $\theta$ is a stabilizing policy in 
$\Theta_S(\lambda)$. 
Then, $\overline{\nu}^{\mathscr{Y}(\theta)} = 
\nu^\theta = \lambda$. 
\end{lemma}

\begin{proof}
First, note 
\begin{align*}
\overline{\nu}^{\mathscr{Y}(\theta)}
& \overset{(a)}{=} \sum_{\overline{\by} \in \mathbb{Y}} 
	\mu(\overline{s}) \ \overline{\pi}^{\mathscr{Y}(\theta)}	
		(\overline{\by}) \ \phi(\overline{\by}) \\
& \overset{(b)}{=} \sum_{\overline{\by} \in \mathbb{Y}}
	\mu(\overline{s}) 
	\Big( \sum_{q\in\mathbb{Q}^{\overline{\by}}} \pi^{\theta}		
		(\overline{\by},q) \Big) \theta(\overline{\by},q) \\
& \overset{(c)}{=} \sum_{\bx\in\mathbb{X}} \mu(s)\pi^{\theta}(\mathbf{x}) \theta(\mathbf{x}) \overset{(d)}{=} \nu^{\theta},
\end{align*} where $(b)$ follows from 
Lemma~\ref{lemma:PMFequivalence}, 
and $(c)$ results from rearranging the summations 
in terms of $\bx = (s,w,q)$. Finally $(a)$ and $(d)$ hold by 
definition. The lemma follows from Remark~
\ref{nuislambda} that $\nu^{\theta}$ is equal to $\lambda$.
\end{proof}

As $\mathscr{Y}(\theta)$ belongs to $\Phi^+_R$ as explained 
earlier, by Lemma~\ref{lem:lambdaisnu},
\beqa
\lambda = \overline{\nu}^{\mathscr{Y}(\theta)}
	\leq \max_{\phi \in \Phi_R} \overline{\nu}^\phi
	\Eqdef \onu^{**}. 
	\label{eq:lambdalessrandom}
\eeqa

\noindent {\bf Step 2:} We shall prove that 
$\onu^* = \onu^{**}$ in two steps. 
First, we establish that there is a stationary 
{\em deterministic} policy that achieves 
$\onu^{**}$. Then, we show that,
for any stationary deterministic policy, we can find
a deterministic {\em threshold} policy that achieves
the same service rate, thereby completing
the proof of Theorem~\ref{thm:necessity}. 

Let us define $\Phi_D$ 
to be a subset of $\Phi_R$, which consists only of  
stationary {\em deterministic} policies for $\overline{\bY}$. 
In other words, if $\phi
\in \Phi_D$, then $\phi(\by) \in \{0, 1\}$ for all 
$\by \in 
\mathbb{Y}$. 
Theorem 9.1.8 in \cite[p. 451]{Puterman2005Markov-decision} 
tells us that if (i) the state space is finite 
and (ii) the set of admissible
actions is finite for every state, there exists a
deterministic stationary optimal policy. Thus,
\begin{equation}
\label{eq:lambdalessdeterm}
\onu^{**} = \max_{\phi \in \Phi_R} \overline{\nu}^{\phi}
=\max_{\phi \in \Phi_D} 
 	\overline{\nu}^{\phi}.
\end{equation} 

While the equality in (\ref{eq:lambdalessdeterm}) 
simplifies the computation of the maximum 
service rate achievable by 
some $\phi$ in $\Phi_R$, 
it requires a search over a set of $2^{n_s}$
deterministic policies in the worst case. Thus, 
when $n_s$ is large, it can be computationally 
expensive. As we
show shortly, it turns out that 
the maximum service rate
on the RHS of (\ref{eq:lambdalessdeterm})
can always be achieved by a
deterministic {\em threshold} policy of the form 
in (\ref{eq:ThresholdPolicy}).

\begin{defn} Recall from (\ref{eq:ThresholdPolicy})
that, for a given $\tau$ in $\mathbb{S} \cup \{n_s+1\}$, 
$\phi_{\tau}$ is the following deterministic threshold 
policy for $\overline{\bY}$:
\begin{equation*}
\phi_{\tau}(\overline{\by}) 
= \begin{cases} 
	0 & \text{if $\overline{s} \geq \tau$ 
		and $\overline{w}=\mathcal{A}$} \\ 
1 & \text{otherwise} \end{cases}
\end{equation*} 
\end{defn} 

The following lemma shows that, for each deterministic 
policy $\phi$ satisfying ${\phi(1,\mathcal{A})=1}$,
there is a deterministic {\em threshold} policy with 
the same service rate.
%\myskip

\begin{lemma}	\label{lemma:perfthreshequal}
Suppose that $\phi$ is a policy in $\Phi_D$ with 
$\phi(1, \cA) = 1$. Then,
$\overline{\nu}^{\phi} = \overline{\nu}^{\phi_{\tau'}}$,
where $\tau' = \mathcal{T}(\phi)+1$.
\end{lemma}
\begin{proof}
We begin with the following facts that will be 
utilized in the proof. 

\begin{itemize}
\item[{\bf F1.}] The postulation that the server is 
non-preemptive, which we formally impose in 
(\ref{ActionConstraintsauxiliary}), means that after 
the sever initiates work on a task, it will be allowed 
to rest only after the task is completed. 
This implies that any policy $\phi$ in $\Phi_D$ 
satisfies $\phi(\overline{s},\mathcal{B})=1$ for all 
$\overline{s}$ in $\mathbb{S}$.

\item[{\bf F2.}] From (\ref{Def-SDynamics}) and (\ref{SDef}), 
we know that $S_{k+1}$ is never smaller than $S_{k}$ if 
the server works during epoch $k$.
\end{itemize}

From {\bf F1} and {\bf F2} stated above, we conclude that 
the following holds for any $\sigma$ in $\mathbb{S}$:
\begin{equation} \label{pushright} 
\phi(\sigma,\mathcal{A})=1 
\implies \Pr(\overline{S}^{\phi}_{k+1} \geq \sigma 
	\ | \ \overline{S}^{\phi}_k = \sigma) 
	= 1
\end{equation} 
Here, we recall that $\overline{\bY}_k^{\phi} 
= (\overline{S}_k^{\phi},\overline{W}_k^{\phi})$ 
represents the state of $\overline{\bY}^{\phi}$ 
at epoch $k$. 

The implication in (\ref{pushright}) leads us to the
following important observation: suppose that a 
deterministic policy $\phi$ in $\Phi_D$ satisfies
$\phi(\sigma, \mathcal{A}) = 1$ for some $\sigma$ 
greater than $1$. Then, all states $(\overline{s},
\overline{w})$ with $\overline{s}$ less than 
$\sigma$ are transient and, therefore, 
\beqa
\label{eq:transientstate}
\overline{\pi}^{\phi}(\overline{s},\overline{w}) = 0
\ \mbox{ if } \ \overline{s} < \sigma. 
\eeqa

The reason for this is 
that (i) because $\phi(1, \cA) =1$, all states
$(\os, \ow)$ with $\os < \sigma$ communicate 
with every state $(\os', \ow')$ with $\os'
\geq \sigma$, 
and (ii) none of the states $(\os', \ow')$
with $\os' \geq \sigma$ communicates with 
any state $(\os, \ow)$ with $\os < \sigma$
since $\phi(\sigma, \cA) = 
\phi(\sigma, \cB) = 1$. 
\myskip

\begin{itemize}
\item[{\bf F3.}]
The above observation means that, given 
a deterministic policy $\phi$ in $\Phi_D$, every 
state $(\os, \ow)$ with $\os < \mathcal{T}(\phi)$
is transient and $\opi^\phi(\os, \ow) = 0$. 

\item[{\bf F4.}]
Moreover, the remaining states $(\os, \ow)$ in 
$\mathbb{Y}^\phi$ with $\os \geq \mathcal{T}(\phi)$
communicate with each other and their period
is one (because it is possible to transition from
$(\mathcal{T}(\phi), \cA)$ to itself). Since 
$\mathbb{Y}^\phi$ is finite, it forms an aperiodic
PRCC of $\overline{\bY}^\phi$.

\end{itemize}

We will complete the proof of Lemma
\ref{lemma:perfthreshequal}
with the help of following lemma. 

\begin{lemma}	\label{lemma:EquivSD}
Suppose that $\phi$ and $\tilde{\phi}$ are two
deterministic policies in $\Phi_D$ satisfying
$\phi(1, \cA) = \tilde{\phi}(1, \cA) = 1$. Then, 
\beqa
\label{eqTimpleqpi}
\mathcal{T}(\tilde{\phi}) = \mathcal{T}(\phi) 
\implies \overline{\bpi}^{\tilde{\phi}} 
	= \overline{\bpi}^{\phi}
\eeqa 
\end{lemma}
\begin{proof}
If ${\mathcal{T}(\tilde{\phi}) = 
\mathcal{T}(\phi)}$, {\bf F3} states that, 
for any state $(\os, \ow)$ 
with $\os < \mathcal{T}(\phi)$, we have 
$\opi^\phi(\os, \ow) = \opi^{\tilde{\phi}}(\os, \ow)
= 0$. Furthermore, {\bf F4} tells us that the 
PRCCs are identical, 
i.e., $\mathbb{Y}^\phi = \mathbb{Y}^{\tilde{\phi}}$. 
From {\bf F1} and the definition of 
mapping $\mathcal{T}$, we conclude that, for all 
states $(\os, \ow)$ in $\mathbb{Y}^\phi$, 
$\tilde{\phi}(\os, \ow) = \phi(\os, \ow)$. 
This in turn means that, for all $(\os, \ow)$
in $\mathbb{Y}^\phi$, we have 
$\overline{\pi}^{\phi}(\os, \ow)
= \overline{\pi}^{\tilde{\phi}}(\os, \ow)$.
\end{proof}

Let us continue with the proof of Lemma
\ref{lemma:perfthreshequal}. 
Select $\tilde{\phi}=\phi_{\tau'}$ with 
$\tau' = \mathcal{T}(\phi)+1$. Then, 
Lemma~\ref{lemma:EquivSD} tells us that 
$\overline{\bpi}^\phi = \overline{\bpi}^{\tilde{\phi}}$. 
From the 
definition of $\overline{\nu}^\phi$ in (\ref{eq:onu}),
Lemma~\ref{lemma:perfthreshequal} 
is now a direct consequence of 
this observation and ${\bf F4}$. 
\end{proof}

The proof of Theorem~\ref{thm:necessity} will be 
completed with the help of
the following intermediate result. It tells us
that we can focus only on the deterministic policies
$\phi$ with $\phi(1, \cA) = 1$.  
%\myskip

\begin{lemma}	\label{lemma:condition1}
There exists a deterministic policy $\phi^*$ with
$\phi^*(1, \cA) = 1$, whose service rate
equals $\onu^{**}$. 
\end{lemma}
\begin{proof}
A proof can be found in Appendix~\ref{appen:condition1}. 
\end{proof}

Proceeding with the proof of the theorem, 
Lemma~\ref{lemma:condition1} tells us that some
deterministic policy $\phi^*$ with $\phi^*(1, \cA) = 1$
achieves the service rate equal to 
$\onu^{**}$. This, together with Lemma
\ref{lemma:perfthreshequal}, proves that 
there exists a deterministic {\em threshold} policy that
achieves $\onu^{**}$ and, as a result, we must have 
$\onu^* = \onu^{**}$.

\subsection{Proof of Theorem \ref{thm:sufficiency}}
\label{subsec:suf}

As mentioned in Section~\ref{subsec:Related}, 
existing results, such as \cite{Baccelli94}, 
can be used to construct a proof of the theorem. Here, 
we provide a more direct proof of the theorem by 
contradiction: suppose that 
the theorem is false and there exists an arrival rate
$\lambda_1 < \onu^*$ for which the system is not 
stable under the policy $\theta_{\tau^*}$. We demonstrate
that this leads to a contradiction. 

For notational convenience, we denote the unique 
stationary distribution of 
$\overline{\bY}^{\phi_{\tau^*}}$ by 
$\overline{\boldsymbol{\pi}}^{\tau^*}$. In addition, 
for each $\bx_0 \in \mathbb{X}$, we define two
sequences of distributions $\{\boldsymbol{\xi}_k^{\bx_0} 
: k \in \N \}$ and 
$\{\boldsymbol{\wp}_k^{\bx_0} : 
k \in \N \}$, where $\boldsymbol{\xi}_k^{\bx_0}$
and $\boldsymbol{\wp}_k^{\bx_0}$ are the distribution of the 
server state $\bY_k$ and the system $\bX_k$, 
respectively, 
at epoch $k \in \N$ under the policy $\theta_{\tau^*}$, 
conditional on $\{\bX_0 = \bx_0\}$.  

We make use of the following lemma to complete the
proof of the theorem. 

\begin{lemma} \label{lemma:DistrConvg1}
Suppose that the system is not stable under the
policy $\theta_{\tau^*}$. Then, for every 
$\varepsilon>0$ and initial state $\bx_0
\in \mathbb{X}$, 
there exists finite $T(\varepsilon, \bx_0)$ 
such that, for all $k \geq T(\varepsilon, \bx_0)$, 
we have 
$\big\Vert \boldsymbol{\xi}_k^{\bx_0} 
	- \overline{\boldsymbol{\pi}}^{\tau^*} 
		\big\Vert_1 
< \varepsilon$.
\end{lemma}
\begin{proof}
A proof is provided in 
Appendix~\ref{appen:DistrConvg1}.
\end{proof}
Proceeding with the proof of the theorem, for every
$k$ in $\N$, 
\beqa
&& \myhb \Pr\Big( D^{\theta_{\tau^*}}_k = 1 \ | \ \bX^{\theta_{\tau^*}}_0 = \bx_0\Big)
= \sum_{\bx \in \mathbb{X}} 
	\theta_{\tau^*}(\bx) \ \mu(s)
		\ \boldsymbol{\wp}_k^{\bx_0}(\bx) \lb
\myeq \sum_{\by \in \mathbb{Y}} 
	\phi_{\tau^*}(\by) \ \mu(s) 
		\Big( \boldsymbol{\xi}^{\bx_0}_k(\by) 
	- \boldsymbol{\wp}_k^{\bx_0}(\by, 0) \Big).
    \label{eq:Suff-00} 
\eeqa
First, since the system is assumed to be not 
stable, Lemma~\ref{lemma:DistrConvg1} tells us
$\lim_{k \to \infty} \boldsymbol{\xi}_k^{\bx_0} 
= \overline{\boldsymbol{\pi}}^{\tau^*}$ for all 
$\bx_0$ in $\mathbb{X}$. 
Second, one can argue that the MDP $\bX^{\theta_{\tau^*}}$
is irreducible because the state $(1, \cA, 0)$
communicates with all other states, and vice versa. 
In addition, $(1, \cA, 0)$ is aperidic because
the probability of transitioning from $(1, \cA, 0)$
to itself is positive. Therefore, all states are either
null recurrent or transient if the system is not stable.   
Since $| \mathbb{Y} |
= 2 n_s$ is finite, this means $\bPP{\theta_{\tau^*}}
{ Q_k = 0 \ | \ \bX_0 = \bx_0} = \sum_{\by \in 
\mathbb{Y}} \boldsymbol{\wp}_k^{\bx_0}(\by, 0)$ 
converges to zero as $k$ goes to $\infty$. 
Thus, for all $\bx_0 \in 
\mathbb{X}$, the probability in (\ref{eq:Suff-00}) 
converges to  
\beqa
\sum_{ \by \in \mathbb{Y} } 
	\phi_{\tau^*}(\by)
	\ \mu(s) \ \overline{\boldsymbol{\pi}}^{\tau^*}(\by)
	= \onu^* \ \mbox{ as } k \to \infty. 
	\label{eq:Suff-0} 
\eeqa
Making use of the reverse Fatou's lemma~\cite[pp. 86-87]{royden1968real}
and the convergence in (\ref{eq:Suff-0}), given any 
initial distribution of $\bX_0$, we obtain
\beqan
&& \hspace{-0.4in} 
\EX\Big[\limsup_{k \to \infty} \
	\frac{1}{k}\sum_{\tau=0}^{k-1} 
		D^{\theta_{\tau^*}}_k \Big] 
\geq \limsup_{k \to \infty} \ \EX\Big[\frac{1}{k}\sum_{\tau=0}^{k-1}
	D^{\theta_{\tau^*}}_k \Big] \lb
\myeq \limsup_{k \to \infty} \ \frac{1}{k}
	\sum_{\tau=0}^{k-1} \EX\big[ D^{\theta_{\tau^*}}_k \big] 
= \onu^*. 
\eeqan
This implies that, for $\delta \Eqdef 0.5 
(\onu^* - \lambda_1) > 0$, we must have 
\beqa
\Pr\Big( \limsup_{k \to \infty} \ \frac{1}{k}
	\sum_{\tau=0}^{k-1} D^{\theta_{\tau^*}}_k > \lambda_1 + \delta\Big)
	> 0. 
	\label{eq:Suff-1}
\eeqa

On the other hand, for all $k \in \N$, the number of 
completed tasks up to epoch $k$ cannot exceed the sum
of the initial queue size $Q_0$ and the number of 
arrivals up to epoch $k$. Thus,  
\beqan
\limsup_{k \to \infty} \ \frac{1}{k}
	\sum_{\tau=0}^{k-1} D^{\theta_{\tau^*}}_k
\myleq \limsup_{k \to \infty} \left( \frac{1}{k}
	\sum_{\tau=0}^{k-1} D^{\theta_{\tau^*}}_k + \frac{Q_0}{k} \right) \lb 
\myleq \lambda_1 \ \mbox{ with probability 1,} 
\eeqan
where the second inequality follows from the
strong law of large numbers.
Clearly, this contradicts the earlier inequality
in (\ref{eq:Suff-1}).

\vspace{-.1 in}
\section{Conclusion}
	\label{sec:Conclusion}
	
We investigated the problem of designing a task 
scheduler policy when the efficiency of the server
is allowed to depend on the past utilization, which is
modeled using an internal state of the server. 
First, we proposed a new framework for studying 
the stability of the queue length of the system. 
Second, making use of the new framework, 
we characterized the set of task arrival rates
for which there exists a stabilizing stationary
scheduler policy. Moreover, finding this set 
can done by solving a simple optimization problem 
over a finite set. Finally, we identified an 
optimal threshold policy that stabilizes the
system whenever the task arrival rate lies 
in the interior of the aforementioned set for
which there is a stabilizing policy.

%APPENDIX%%%%%%%%%%%%%%%%%%%%%%%%%%%%%%%%%%%%%%%%%%%%%%%%%%%%%%%%%%%%%%%%%%%%%%%%%%%%%%%
%\appendix 
\begin{appendices}

\section{Proof of Lemma \ref{lemma:UniquePMF}}
	\label{appen:UniquePMF}

We will prove the claim by contradiction. The
decomposition theorem of MCs tells us that
$\mathbb{X}$ can be partitioned into a set 
consisting of transient states and a collection 
of {\em irreducible}, {\em closed} RCCs 
$\{\mathbb{C}^1, 
\mathbb{C}^2, \ldots\}$
\cite[p. 224]{Grimmett2001Probability-and}. Since
$\bX^\theta$ is assumed stable, all RCCs 
$\mathbb{C}^m$, $m = 1, 
2, \ldots$, are positive recurrent. 
Suppose that the claim is false and there is
more than one PRCC. We demonstrate that this 
leads to a contradiction. 

First, we show that $\mathbb{C}^m$, $m = 1, 2, 
\ldots$, include a state $(s_m, \cB, q_m)$ for 
some $s_m \in \mathbb{S}$ and $q_m > 0$. If this 
is not true, every state in $\mathbb{C}^m$ is of 
the form $(s, \cA, q)$ and $\theta(s, \cA, q) = 
0$ because $\mathbb{C}^m$ is closed~\cite{Grimmett2001Probability-and}. 
But, this implies that, starting with any state 
in $\mathbb{C}^m$, the scheduler will never
assign a task to the server and, consequently, 
all states in $\mathbb{C}^m$ must be transient, 
which contradicts that $\mathbb{C}^m$ is positive
recurrent. For the same reason, each $\mathbb{C}_m$
must include a state $\tilde{\bx}_m = (\tilde{s}_m, 
\cA, \tilde{q}_m)$ with $\theta(\tilde{\bx}_m) > 0$, 
which implies that $\mathbb{C}_m$ is aperiodic. 

Second, if some state $(s, \cB, q)$ is in 
$\mathbb{C}^m$, $m = 1, 2, \ldots$, 
then so are all the states $(s', w, q')$ for
all $s' \geq s$, $w$ in $\mathbb{W}$, and $q'
\geq q$: the fact that $(s, \cB, q)$ communicates
with $(s', w, q)$, $s' \geq s$ and $w$ in 
$\mathbb{W}$, which means that these states 
belong to $\mathbb{C}^m$ as well, 
is obvious. In addition, 
it is evident that $(s', \cB, q)$ communicates
with $(s', \cB, q')$ for all $q' \geq q$. 
In order to see why $(s', \cA, q')$, $s' \geq s$ 
and $q' \geq q$, also lie in $\mathbb{C}^m$, 
consider the following two cases: if
$\theta(s', \cA, q) = 0$, then clearly
$(s', \cA, q)$ communicates with $(s', \cA,
q+1)$. On the other hand, if $\theta(s', \cA, q)
> 0$, $(s', \cA, q)$ communicates with 
$(s', \cB, q+1)$, which in turn communicates
with $(s', \cA, q+1)$. The claim now follows
by induction.

Note that, the above two observations together imply
that there exists finite $q^* \mydef \max\{q_1, 
q_2\}$ such that all states $(n_s, w, q)$, $w$ in 
$\mathbb{W}$ and $q \geq q^*$, belong to both 
$\mathbb{C}^1$ and $\mathbb{C}^2$. This, however,
contradicts the earlier assumption that 
$\mathbb{C}^1$ and $\mathbb{C}^2$ are {\em 
disjoint} PRCCs.

\vspace{-.1 in}
\section{Proof of Lemma~\ref{lemma:PhiThetaPos}}
	\label{appen:PhiThetaPos}

We start with two observations: first, recall from 
the decomposition 
theorem~\cite{Grimmett2001Probability-and} that 
$\mathbb{C}_\theta$ is closed.
Thus, once $\mathbf{X}^\theta$
reaches a state in $\mathbb{C}_\theta$, it
never leaves $\mathbb{C}_\theta$.
Second, if $(1, {\cal A}, q') =: {\bf x}' \in 
\mathbb{C}_\theta$
for some $q' \geq 1$, then $(1, {\cal A}, q)
\in \mathbb{C}_\theta$ for all $q \geq q'$.
This can be shown by considering two cases: (a)
$\theta({\bf x}') > 0$ and (b) $\theta({\bf x}') 
= 0$. In case (a), consider the following
event with positive probability: suppose
${\bf X}_k^\theta = {\bf x}'$ at epoch $k$. 
The server assigns a new task at epoch $k$, 
there is an arrival at each epoch $k, k+1, 
\ldots, k + q - q'$, and the server
completes the service of the task assigned at
epoch $k$
during epoch $k + q - q'$. When this
event occurs, we have ${\bf X}^\theta_{k + q 
- q' + 1} = (1, {\cal A}, 
q)$. In case (b), define
$q^* = \min\{ \tilde{q} > q' \ | \ \theta(1, 
{\cal A}, \tilde{q}) 
> 0 \}$ with the convention $q^* = \infty$ if
the set is empty. If $q^* \geq q$,
${\bf x}'$ communicates with $(1, {\cal A}, q)$. 
If $q^* < q$, ${\bf x}'$
communicates with $(1, {\cal A}, q^*)$, which 
in turn communicates with $(1, {\cal A}, q)$
by the same argument used in the first case. 
Since $\mathbb{C}_\theta$ is closed, it follows
$(1, {\cal A}, q) \in \mathbb{C}_\theta$ 
and, hence, $\pi^\theta(1, {\cal A}, q) > 0$
for all $q \geq q'$ \cite{Grimmett2001Probability-and}.

We now prove that there is
some $q^+ \geq 1$ such that $(1, {\cal A} , q^+)
\in \mathbb{C}_\theta$. Suppose that this is
false. Then, $\mathbb{C}_\theta \subseteq
\{(s, {\cal W}, q) \ | \ s \in \mathbb{S}, \
q \in \N\} \cup \{(1, {\cal A}, 0)\}$. 
This implies that the server does not
schedule new tasks for service at steady 
state with distribution 
$\boldsymbol{\pi}^\theta$ and obviously 
contradicts the 
assumption $\mathbf{X}^\theta$ is stable. 
Finally, $\sum_{q=q^+}^{\infty} \theta(1, \cA, 
q)$ must be strictly positive when $\theta$ is a 
stabilizing policy; otherwise, states 
$(1, {\cal A}, q)$, $q \geq q^+$, would be
transient instead. The lemma now follows from an 
earlier observation $\pi^\theta(1, {\cal A}, 
q) > 0$ for all $q \geq q^+$ and 
the definition of $\phi^\theta(1, {\cal A})$ in 
\eqref{QMapDef}.

%BALANCE EQUATION%
\section{Proof of Lemma~\ref{lemma:PMFequivalence}}
	\label{appen:PMFequivalence}

For notational convenience, let $\phi = 
\mathscr{Y}(\theta)$. 	
Taking advantage of the fact that there is a unique 
stationary PMF of $\overline{\bY}^{\phi}$, 
it suffices to show that the distribution given in 
(\ref{eq:pi-relation}) satisfies the definition of 
stationary PMF:
\beqa
\opi^\phi(\oby)
\myeq \sum_{\oby' \in \mathbb{Y}} \opi^\phi(\oby') \ 
	\overline{\bf P}^\phi_{\oby', \oby} \ 
	\mbox{ for all } \oby \in \mathbb{Y}, 
	\label{eq:lemma1-1}
\eeqa
where $\overline{\bf P}^\phi$ denotes the one-step transition matrix of $\overline{\bY}^{\phi}$.

\noindent $\bullet$ 
{\bf Right-hand side of (\ref{eq:lemma1-1}):} 
Using the policy $\phi$ in place,
\beqa
&& \myhb \sum_{\oby' \in \mathbb{Y}} \opi^\phi(\oby') \ 
	\overline{\bf P}^\phi_{\oby', \oby}  \lb
\myeq \sum_{\oby' \in \mathbb{Y}} \opi^\phi(\oby')
	\big( \phi(\oby') \obP^{\mathcal{W}}_{\oby', \oby}
		+ (1 - \phi(\oby')) \obP^{\mathcal{R}}_{\oby', \oby} 
			\big) \lb 
\myeq \sum_{\oby' \in \mathbb{Y}} \opi^\phi(\oby')
	\Big( \phi(\oby') (\obP^{\mathcal{W}}_{\oby', \oby}
	- \obP^{\mathcal{R}}_{\oby', \oby} \big) 
	+ \obP^{\mathcal{R}}_{\oby', \oby} \Big), 
	\label{eq:lemma1-2}
\eeqa
where $\obP^{\mathcal{R}}$ (resp. $\obP^{\mathcal{W}}$)
denotes the one-step transition matrix of $\obY$
under a policy that always rests (resp. works on a 
new task) when available.

Substituting (\ref{QMapDef}) for $\phi(\oby')$ 
in (\ref{eq:lemma1-2}) and using the given expression
$\opi^\phi(\oby') = \sum_{q \in \mathbb{L}^{\oby'}}
\pi^\theta(\oby', q)$ in (\ref{eq:pi-relation}), 
we obtain
\beqa
&& \myhb (\ref{eq:lemma1-2})
= \sum_{\oby' \in \mathbb{Y}} \sum_{q' \in \mathbb{L}^{\oby'}}
	\theta(\oby', q') \pi^\theta(\oby', q') 
	\big( \obP^{\mathcal{W}}_{\oby', \oby}
	- \obP^{\mathcal{R}}_{\oby', \oby} \big) \lb 
&& \hspace{0.2in} 
	+ \sum_{\oby' \in \mathbb{Y}} 
		\sum_{q \in \mathbb{L}^{\oby'}} \pi^\theta(\oby', q)
	\obP^{\mathcal{R}}_{\oby', \oby} \lb 
\myeq \sum_{\bx' \in \mathbb{X}} \theta(\bx') 
	\pi^\theta(\bx') \big( \obP^{\mathcal{W}}_{\oby', \oby}
	- \obP^{\mathcal{R}}_{\oby', \oby} \big) 
	+ \sum_{\bx' \in \mathbb{X}} \pi^\theta(\bx') 
	\obP^{\mathcal{R}}_{\oby', \oby}. \lb
	\label{eq:lemma1-3}
\eeqa

\noindent $\bullet$
{\bf Left-hand side of (\ref{eq:lemma1-1}):} Using (\ref{eq:pi-relation}), we get
\beqa
\opi^\phi(\oby)
\myeq \sum_{q \in \mathbb{L}^{\oby}} \pi^\theta(\oby, q). 
	\label{eq:lemma1-4}
\eeqa
For notational simplicity, we denote $(\oby, q)$ on the RHS
of (\ref{eq:lemma1-4}) simply by $\bx$. 
Since $\bpi^\theta$ is the unique stationary PMF of 
$\bX^\theta$, we have 
\beqa
&& \myhb \pi^\theta(\oby, q) 
= \sum_{\bx' \in \mathbb{X}} \pi^\theta(\bx')
	{\bf P}^\theta_{\bx', \bx}
	\label{eq:lemma1-5} \\
\myeq \sum_{\bx' \in \mathbb{X}} \pi^\theta(\bx')
	\theta(\bx') \big( {\bf P}^{\mathcal{W}}_{\bx', \bx}
		- {\bf P}^{\mathcal{R}}_{\bx', \bx} \big) 
	+ \sum_{\bx' \in \mathbb{X}} \pi^\theta(\bx')
		{\bf P}^{\mathcal{R}}_{\bx', \bx},  
	\nonumber
\eeqa
where ${\bf P}^\theta$ is the 
one-step transition matrix of $\bX^\theta$, and 
${\bf P}^{\mathcal{R}}$ (resp. ${\bf P}^{\mathcal{W}}$)
is the one-step transition matrix under a policy 
that always rests (resp. assigns a new task) when 
the server is available and at least one task is 
waiting for service. 

Substituting (\ref{eq:lemma1-5}) in (\ref{eq:lemma1-4})
and rearranging the summations, we obtain 
\beqa
\opi^{\phi}(\oby)
\myeq \sum_{\bx' \in \mathbb{X}} \pi^\theta(\bx')
	\theta(\bx') \sum_{q \in \mathbb{L}^{\oby}}
	\big( {\bf P}^{\mathcal{W}}_{\bx', \bx}
		- {\bf P}^{\mathcal{R}}_{\bx', \bx} \big) \lb 
&& + \sum_{\bx' \in \mathbb{X}} \pi^\theta(\bx')
	\sum_{q \in \mathbb{L}^{\oby}} 
		{\bf P}^{\mathcal{R}}_{\bx', \bx}.
	\label{eq:lemma1-6}
\eeqa

From (\ref{eq:lemma1-3}) and (\ref{eq:lemma1-6}), 
in order to prove (\ref{eq:lemma1-1}), it suffices to 
show
%\beqan
$\obP^{a}_{\oby', \oby}
= \sum_{q \in \mathbb{L}^{\oby}} 
	{\bf P}^a_{(\oby', q'), (\oby, q)}, 
	\ a \in \mathbb{A}$. 
%\eeqan
Note that
\beqa
\myhb \sum_{q \in \mathbb{L}^{\oby}} 
	{\bf P}^a_{(\oby', q'), (\oby, q)} 
\myeq P_{\bY_{k+1} | (\bY_k, Q_k), A_k}( \oby
	\ | \ (\oby', q'), a). 
	\label{eq:lemma1-7}
\eeqa
Clearly, conditional on $\{ (\bY_k, A_k) = (\oby', a) \}$, 
$\bY_{k+1}$ does not depend on the queue size
at epoch $k$. As a result, the RHS of (\ref{eq:lemma1-7})
does not depend on $q'$ and is equal to 
$P_{\bY_{k+1} | \bY_k, A_k}( \oby
	\ | \ \oby', a) = \obP^{a}_{\oby', \oby}$.

\section{Proof of Lemma~\ref{lemma:condition1}}
	\label{appen:condition1}
	
In order to prove the lemma, it suffices to show
the following: suppose that $\phi'$ is a deterministic 
policy with $\phi'(1, \cA) = 0$ and achieves 
$\onu^{**}$. 
Then, we can find another deterministic policy $\phi^*$
with $\phi^*(1, \cA) = 1$ which achieves $\onu^{**}$. 

Suppose $\phi'$ satisfies $\phi'(1, \cA) = 0$
and $\overline{\nu}^{\phi'} = \onu^{**}$. 
First, we can show that $\mathcal{T}(\phi') > 1$ by 
contradiction: assume $\mathcal{T}(\phi') = 0$. 
Then, from the discussion in 
Section~\ref{subsec:PMFobY}, 
we know that \{$(1, \cA)$\} is the unique PRCC of 
$\overline{\bY}^{\phi'}$
and, hence, 
$\overline{\nu}^{\phi'} = 0$, thereby contracting
the earlier assumption that $\overline{\nu}^{\phi'}
= \onu^{**} > 0$. 

Since $\mathcal{T}(\phi') > 1$, we know that there 
are two PRCCs
-- $\{(1, \cA)\}$ and $\mathbb{Y}^{\phi'}$ --
and all other states are transient.
But, the first PRCC does not contribute to the long-term service
rate and, as a result, the service rate 
$\overline{\nu}^{\phi'} = \onu^{**}$ 
is achieved by a 
stationary PMF only over $\mathbb{Y}^{\phi'}$, 
which we denote by $\breve{\bpi}^{\phi'}$. 

Consider a new deterministic policy $\phi^+$ with
\beqan
\phi^+(\os, \ow)
\myeq \begin{cases}
	1 & \mbox{if } 1 \leq \os \leq \mathcal{T}(\phi'),  \\
	\phi'(\os, \ow) & \mbox{otherwise.} 
	\end{cases}
\eeqan
By construction, clearly $\mathcal{T}(\phi') 
= \mathcal{T}(\phi^+)$ and, hence, 
$\mathbb{Y}^{\phi'} = \mathbb{Y}^{\phi^+}$. 
Because $\phi'(\os, \ow) = \phi^+(\os, \ow)$ 
if $\os \geq \mathcal{T}(\phi')$, it follows that
$\overline{\bpi}^{\phi^+} = \breve{\bpi}^{\phi'}$
and, as a result, $\overline{\nu}^{\phi'}
= \overline{\nu}^{\phi^+} = \onu^{**}$.

\section{Proof of Lemma~\ref{lemma:DistrConvg1}}
	\label{appen:DistrConvg1}
	
Denote the one-step transition matrix of 
$\overline{\bY}^{\phi_{\tau^*}}$ 
by $\overline{\bf P}^{*}$. 
We prove Lemma~\ref{lemma:DistrConvg1} 
with the help of the following
two lemmas. 

\begin{lemma}	\label{lemma:DistrConvg2} 
For every $\varepsilon > 0$, there exists finite
$T_1(\varepsilon)$ such 
that, for all $k \geq T_1(\varepsilon)$
and any distribution ${\bf p}$ 
over $\mathbb{Y}$, we have 
$\big\Vert {\bf p} \ 
	\big({\overline{\bf P}^{*}} \big)^k 
	- \overline{\bpi}^{\tau^*} 
		\big\Vert_1 < 0. 5 \varepsilon$.
\end{lemma}
\begin{proof}
Let $\overrightarrow{{\bf P}}^{*} \Eqdef
\lim_{k \to \infty} 
\big(\overline{\bf P}^{*} \big)^k$, 
whose rows are equal to 
$\overline{\bpi}^{\tau^*}$. Then, given 
any distribution ${\bf p}$ over 
$\mathbb{Y}$, we have ${\bf p} 
\overrightarrow{{\bf P}}^{*} = 
\overline{\bpi}^{\tau^*}$. Using this 
equality, for all sufficiently large
$k$, we have 
\beqan
&& \myhb \left\Vert {\bf p} 
	\big( \overline{\bf P}^{*} \big)^k 
		- \overline{\bpi}^{\tau^*}  
			\right\Vert_1
= \left\Vert {\bf p} 
	\big( \overline{\bf P}^{*} \big)^k 
		- {\bf p} 
		\overrightarrow{{\bf P}}^{*}
			\right\Vert_1 \lb
\myleq \left\Vert {\bf p} \right\Vert_1  
	\cdot \left\Vert \big( \overline{\bf P}^{*} \big)^k 
		- \overrightarrow{{\bf P}}^{*}
			\right\Vert_{\infty}
= \left\Vert \big( \overline{\bf P}^{*} \big)^k 
		- \overrightarrow{{\bf P}}^{*}
			\right\Vert_{\infty} 
    < 0.5 \varepsilon. 
\eeqan
The last inequality follows from the 
definition of $\overrightarrow{{\bf P}}^{*}$. 
\end{proof}

\begin{lemma} 	\label{lemma:DistrConvg3}
Suppose that the system is not stable. Then, 
for all $\varepsilon>0$, positive 
integer $N$, and initial state 
$\bx_0 \in \mathbb{X}$, 
there exists finite $T_2(\varepsilon, 
N, \bx_0)$ such that, for all $k \geq N + 
T_2(\varepsilon, N, \bx_0)$, we have 
$\big\Vert 
	\boldsymbol{\xi}_k^{\bx_0} 
		- \boldsymbol{\xi}_{k-N}^{\bx_0}
		\big( \overline{\bf P}^{*} \big)^N  
			\big\Vert_1 
< 0.5 \varepsilon$.
\end{lemma}
\begin{proof} 
Let $\overline{\bf P}^{\mathcal{R}}$ 
be the one-step
transition matrix of $\overline{\bY}$ 
under policy $\phi_1$ that always chooses 
$\mathcal{R}$ when the server is available.
We denote the row of $\overline{\bf P}^{*}$ 
(resp. $\overline{\bf P}^{\mathcal{R}}$) 
corresponding to the server state $\by 
= (s, w) \in 
\mathbb{Y}$ by $\overline{\bf P}^{*}_{\by}$
(resp. $\overline{\bf P}^{\mathcal{R}}_{\by}$).

By conditioning on $\bX_{k-N}$ and using the 
equality $\boldsymbol{\xi}^{\bx_0}_k(\by)
= \sum_{q \in \mathbb{L}^{\by}} 
\boldsymbol{\wp}^{\bx_0}_{k}(\by, q)$, we can rewrite
$\boldsymbol{\xi}^{\bx_0}_{k-N+1}$ as 
\beqa
&& \hspace{-0.52in} 
\boldsymbol{\xi}^{\bx_0}_{k-N+1} 
%\lb
= \sum_{s \in \mathbb{S}} \Big[
	\boldsymbol{\wp}^{\bx_0}_{k-N}(s, \cA, 0)
			\ \overline{\bf P}^{\mathcal{R}}_{(s,\cA)} \lb
&& \hspace{0.5in}
	+ \sum_{w \in \mathbb{W}} \Big( \sum_{q = 1}^\infty 
	\boldsymbol{\wp}^{\bx_0}_{k-N}(s, w, q)  \Big)
			\overline{\bf P}^*_{(s,w)} \Big] \lb 
&& \hspace{-0.58in} = \sum_{s \in \mathbb{S}}
	\Big[ \boldsymbol{\wp}^{\bx_0}_{k-N}(s, \cA, 0) 
			\ \overline{\bf P}^{\mathcal{R}}_{(s, \cA)} 
	+\boldsymbol{\xi}^{\bx_0}_{k-N}(s, \cB) 
		\ \overline{\bf P}^*_{(s, \cB)} \lb
&& + \Big( \boldsymbol{\xi}^{\bx_0}_{k-N}(s, \cA)
	- \boldsymbol{\wp}^{\bx_0}_{k-N}(s, \cA, 0) \Big)
			\overline{\bf P}^*_{(s, \cA)} \Big] \lb
&& \hspace{-0.58in} =  \sum_{s \in \mathbb{S}} 
	\Big[ \boldsymbol{\wp}^{\bx_0}_{k-N}(s, \cA, 0) 
		\big( \overline{\bf P}^{\mathcal{R}}_{(s, \cA)} 
			- \overline{\bf P}^{*}_{(s, \cA)}  \big)
					\Big] %\lb
%hspace{0.0in} 
+ \boldsymbol{\xi}^{\bx_0}_{k-N} \ 
	\overline{\bf P}^*.
		\label{eq:DC3-1}
\eeqa

Define 
$\boldsymbol{\gamma}^{\bx_0}_{\ell}
\Eqdef \sum_{s \in \mathbb{S}} 
	\big[ \boldsymbol{\wp}^{\bx_0}_{\ell}(s, \cA, 0) 
		\big( \overline{\bf P}^{\mathcal{R}}_{(s, \cA)} 
			- \overline{\bf P}^{*}_{(s, \cA)}  \big)
					\big]$, $\ell \in \N$. 
Applying (\ref{eq:DC3-1}) iteratively, we
obtain 
\beqa
\boldsymbol{\xi}^{\bx_0}_{k}
\myeq \boldsymbol{\xi}^{\bx_0}_{k-N} 
	\big( \overline{\bf P}^* \big)^N 
	+ \sum_{\ell=1}^{N} 
		\boldsymbol{\gamma}^{\bx_0}_{k-\ell}
		\big( \overline{\bf P}^* \big)^{\ell-1}.
	\label{eq:DC3-2}
\eeqa
Subtracting the first term on the RHS
of (\ref{eq:DC3-2}) from both sides and taking
the norm, 
\beqan
&& \myhb \left\Vert \boldsymbol{\xi}^{\bx_0}_{k}
	- \boldsymbol{\xi}^{\bx_0}_{k-N} 
		\big( \overline{\bf P}^* \big)^N  
	 		\right\Vert_1 
= \left\Vert \sum_{\ell=1}^{N} 
		\boldsymbol{\gamma}^{\bx_0}_{k-\ell}
		\big( \overline{\bf P}^* \big)^{\ell-1}
			\right\Vert_1 \lb
%\myleq \sum_{\ell=1}^{N} \left\Vert 
%	\boldsymbol{\gamma}^{\bx_0}_{k-\ell}
%		\big( \overline{\bf P}^* \big)^{\ell-1}
%			\right\Vert_1 \lb
\myleq \sum_{\ell=1}^{N} \left\Vert
	\boldsymbol{\gamma}^{\bx_0}_{k-\ell}
		\right\Vert_1 \cdot
	\left\Vert 
		\big( \overline{\bf P}^* \big)^{\ell-1}
			\right\Vert_\infty 
= \sum_{\ell=1}^{N} \left\Vert
	\boldsymbol{\gamma}^{\bx_0}_{k-\ell}
		\right\Vert_1.
	\label{eq:DC3-3}
\eeqan
Substituting the expression for $\gamma^{\bx_0}_{k-\ell}$
and using the inequality 
$\norm{\overline{\bf P}^{\mathcal{R}}_{\by}
- \overline{\bf P}^*_{\by}}_1 \leq 2$ for all 
$\by \in \mathbb{Y}$, we get
\beqa
\sum_{\ell=1}^{N} \left\Vert
	\boldsymbol{\gamma}^{\bx_0}_{k-\ell}
		\right\Vert_1
\myleq 2 \sum_{\ell=1}^{N} \Big( \sum_{s \in \mathbb{S}}
	\boldsymbol{\wp}^{\bx_0}_{k-\ell}(s, \cA, 0) \Big).
	\label{eq:DC3-4}
\eeqa

Recall that if $\bX^{\theta_{\tau^*}}$ is not stable, 
all states are either null recurrent or transient. 
This implies that, for all $\bx_0 \in \mathbb{X}$, 
$\lim_{k \to \infty} \boldsymbol{\wp}^{\bx_0}_{k}(s, \cA, 0)
= 0$ for all $s \in \mathbb{S}$. Therefore, 
for all $\varepsilon > 0$, fixed $N$, and an initial
state $\bx_0$, there exists finite $T_2(\varepsilon, 
N, \bx_0)$ such that, for all $k \geq T_2(\varepsilon, 
N, \bx_0)$, we have $\boldsymbol{\wp}^{\bx_0}_{k}(s, \cA, 0)
< \frac{\varepsilon}{4 N | \mathbb{S} | }$. 
Consequently, the RHS of $(\ref{eq:DC3-4})$ is smaller 
than $0.5 \varepsilon$ for all $k \geq T_2(\varepsilon,
N, \bx_0) + N$. 
\end{proof}

We now proceed with the proof of Lemma
\ref{lemma:DistrConvg1}. 
\beqa
&& \myhb \left\Vert \boldsymbol{\xi}_k^{\bx_0} 
	- \overline{\bpi}^{\tau^*} 
		\right \Vert_1 
	\label{eq:DC1-1} \\
\myleq \left\Vert \boldsymbol{\xi}_k^{\bx_0} 
	- \boldsymbol{\xi}^{\bx_0}_{k-N}
	\big( \overline{\bf P}^* \big)^N
		\right \Vert_1 
	+ \left\Vert \boldsymbol{\xi}^{\bx_0}_{k-N}
	\big( \overline{\bf P}^* \big)^N
	- \overline{\bpi}^{\tau^*} 
		\right \Vert_1 
	\nonumber
\eeqa
Since Lemma~\ref{lemma:DistrConvg2} holds for 
any distribution ${\bf p}$, if we choose
${\bf p} = \boldsymbol{\xi}^{\bx_0}_{k-N}$, 
for all $N \geq T_1(\varepsilon)$, the second term in 
(\ref{eq:DC1-1}) is smaller than $0.5 \varepsilon$. 
In addition, Lemma~\ref{lemma:DistrConvg3} tells us that 
we can find $T_2(\varepsilon, N, \bx_0)$ such that, for all 
$k \geq N + T_2(\varepsilon, N, \bx_0)$, the first term in 
(\ref{eq:DC1-1}) is upper bounded by $0.5 \varepsilon$. 
Thus, it is clear that Lemma~\ref{lemma:DistrConvg1} 
holds with $T(\varepsilon, \bx_0) = 
T_1(\varepsilon) +T_2(\varepsilon, T_1(\varepsilon), \bx_0)$.

\end{appendices}

\bibliographystyle{ieeetr}

\bibliography{MartinsRefs}

\end{document}